\documentclass[journal]{IEEEtran}

\IEEEoverridecommandlockouts	

\usepackage[latin1]{inputenc}
\usepackage{cite,url}

\usepackage[usenames,dvipsnames]{color}
\usepackage{graphicx}

\graphicspath{{./}{fig/}{bio/}{matlab_figs/}}
\DeclareGraphicsExtensions{.pdf,.jpeg,.png,.jpg}

\usepackage{dsfont}
\usepackage{balance}

\usepackage{amssymb,latexsym,epsfig,epic,eepic}
\usepackage[cmex10]{amsmath}
\usepackage{amsthm}
\usepackage{amsfonts}
\usepackage{amssymb}
\usepackage{mathrsfs}
\usepackage{subcaption}
\usepackage{mathtools, bbm}
\usepackage{enumitem}
\setlist[enumerate]{leftmargin=*}
\setlist[itemize]{leftmargin=*}

\newtheoremstyle{bfnote}%
{}{}%
{\itshape}{}%
{\bfseries}{.}%
{ }%
{\thmname{#1}\thmnumber{ #2}\thmnote{ (#3)}}

\theoremstyle{bfnote}
\definecolor{magenta}{RGB}{0,56,167}

\newtheorem{theorem}{Theorem}
\newtheorem{lemma}[theorem]{Lemma}

\newtheorem{proposition}[theorem]{Proposition}

\newtheorem{remark}{Remark}
\newtheorem{example}{Example}

\newtheorem{assumption}{Assumption}

\newcommand{\reeq}[1]{(\ref{eq.#1})}
\newcommand{\ha}{\frac{1}{2}}
\newcommand{\diag}  {\ensuremath {\mathrm{diag}}}
\newcommand{\ip}[2]{\langle{#1},{#2}\rangle}
\newcommand{\Lt}{  \mathcal{L}_2 }
\newcommand{\Ht}{ \mathcal{H}_{2}}
\newcommand{\Hinf}{ \mathcal{H}_{\infty}}
\newcommand{\w}{\omega}

\newcommand{\R}{{\mathbb R}}

\newcommand{\tr}{\mbox{\rm Tr}}

\newcommand{\btmz}{\begin{itemize}}
\newcommand{\etmz}{\end{itemize}}
\newcommand{\benum}{\begin{enumerate}}
\newcommand{\eenum}{\end{enumerate}}

\newcommand{\bmat}[1]{\begin{bmatrix}#1\end{bmatrix}}

\newcommand{\stsp}[4]{
\left[ \begin{array}{c|c}
       #1 & #2 \\ \hline
       #3 & #4
       \end{array} \right]}

\usepackage{ifthen}
\newboolean{showcomments}
\setboolean{showcomments}{true}   
\usepackage{todonotes}

\definecolor{bleudefrance}{rgb}{0.19, 0.55, 0.91}
\definecolor{ao(english)}{rgb}{0.0, 0.5, 0.0}
\newcommand{\enrique}[1]{  \ifthenelse{\boolean{showcomments}}
{\todo[inline,color=bleudefrance,caption={}]{Enrique says: #1}}{}}

\newcommand{\fernando}[1]{  \ifthenelse{\boolean{showcomments}}
{\todo[inline,color=ao(english),caption={}]{Fernando says: #1}}{}}

\newcommand{\addcite}[0]{\ifthenelse{\boolean{showcomments}}
{\textcolor{purple}{(add cite(s)) }}{}}%

\newcommand{\emmargin}[1]{\ifthenelse{\boolean{showcomments}}
{\todo{Enrique: #1)}}{}}
\newcommand{\fpmargin}[1]{\ifthenelse{\boolean{showcomments}}
{\todo{Paganini: #1)}}{}}

\usepackage{changes}
\definechangesauthor[name={Enrique}, color=bleudefrance]{em}
\setremarkmarkup{(#2)}

\newcommand{\at}[1]{
\ifthenelse{\boolean{showcomments}}
{\added{#1}}{#1}
}
\newcommand{\dt}[1]{
\ifthenelse{\boolean{showcomments}}
{\remove{#1}}
{\relax}
}
\newcommand{\ct}[2]{
\ifthenelse{\boolean{showcomments}}
{\changed{#1}{#2}}
{#2}
}

\allowdisplaybreaks

\renewcommand{\baselinestretch}{.98}

\title{\LARGE \bf Global analysis of synchronization performance for
power systems: bridging the theory-practice gap}

\author{Fernando Paganini,~\IEEEmembership{Fellow,~IEEE}, and Enrique Mallada,~\IEEEmembership{Senior Member,~IEEE}
\thanks{Fernando Paganini is with Universidad ORT Uruguay, Montevideo. E-mail: {\tt paganini@ort.edu.uy}. Enrique Mallada is with Johns Hopkins University, Baltimore. E-mail: {\tt mallada@jhu.edu}. This work was supported in part by IDB/MIEM-Uruguay, Project ATN/KF 13883 UR,  ANII-Uruguay, grant FSE\_1\_2016\_1\_131605, NSF, grants CNS 1544771, EPCN 1711188,  AMPS 1736448, and CAREER 1752362, and US DoE EERE award DE-EE0008006.}}

\begin{document}

\maketitle

\begin{abstract}
The issue of synchronization in the power grid is receiving renewed attention, as new energy sources with different dynamics enter the picture. Global metrics have been proposed to evaluate performance, and analyzed under highly simplified assumptions. In this paper we extend this approach to more realistic network scenarios, and more closely connect it with metrics used in power engineering practice.
In particular, our analysis covers networks with generators of heterogeneous ratings and richer dynamic models of machines. Under a suitable proportionality assumption in the parameters, we show that the step response of bus frequencies can be decomposed in two components. The first component is a {system-wide frequency} that captures the aggregate grid behavior, and the residual component represents the individual bus frequency deviations from the aggregate.

Using this decomposition, we define --and compute in closed form-- several metrics that capture dynamic behaviors that are of relevance for power engineers.
In particular, using the \emph{system frequency}, we define industry-style metrics (Nadir, RoCoF) that are evaluated through a representative machine. We further use the norm of the residual component to define a \emph{synchronization cost} that can appropriately quantify inter-area oscillations.
Finally, we employ robustness analysis tools to evaluate deviations from our proportionality assumption. We show that the system frequency still captures the grid steady-state deviation, and becomes an accurate reduced-order model of the grid as the network connectivity grows.

Simulation studies with practically relevant data are included to validate the theory and further illustrate the impact of network structure and parameters on synchronization. Our analysis gives conclusions of practical interest, sometimes challenging the conventional wisdom in the field.
\end{abstract}

\section{Introduction}\label{sec.intro}

The electric power grid, constructed a century ago under the alternating current model, relies strongly on the \emph{synchronization} of the multiple oscillatory variables at a common nominal frequency (e.g. 60Hz in the U.S., 50Hz in Europe).
Automatic control mechanisms are deployed to keep deviations from these set points very small (up to the order of 100mHz). If larger fluctuations are observed, drastic measures are taken: machine protections or load shedding ~\cite{NERC:2015tc} automatically disconnect network elements, with the potential cost of cascading failures and ultimately blackouts~\cite{FederalEnergyRegulatoryCommission:2012wq}.

One challenge for the design of such controls is that frequency measurements are node dependent, and individual nodes see the superposition of two different phenomena: system-wide frequency changes which reflect a global supply-demand imbalance, and  inter-area frequency oscillations ~\cite{Hsu:di, klein_fundamental_1991, messina_inter-area_2011,Stankovic:1999gs} due to weak dynamic coupling. In this regard, frequency regulation controllers for normal operation include a fast, decentralized response by local generators, and a slower set-point correction coordinated by the system operator~\cite{Stankovic:1998ii}. To identify when protections must be triggered due to a sudden fault, power engineers traditionally rely on step response notions such as the  maximum frequency deviation (Nadir) or the maximum rate of change of frequency (RoCoF)\cite{Miller:2011tm}; again these are typically based on node specific measurements.

The analysis of grid synchronization from a \emph{global} perspective has been pursued in academia, going back to eigenvalue methods for reduced order modelling of inter-area oscillations (slow coherency \cite{Winkelman:1981}, participation factors \cite{Verghese:1982}). More recently, researchers from the control field have proposed to use $\mathcal{H}_2$ or $\mathcal{H}_\infty$ norms of global transfer functions~\cite{bassam,mevsanovic2016comparison,m2016cdc,poolla_dorfler2017,simpson-porco2017,andreasson2017,jpm2017cdc,7604066,coletta2017performance,coletta2018transient} as figures of merit for synchronization performance. Ideally, these measures should isolate the effect of different aspects of the network structure (topology, damping, inertia, etc.).  Of particular importance is to determine whether the decreased \emph{inertia} in new systems with renewable generation causes, as is feared \cite{o2014studying}, a deterioration of frequency regulation.

Some model simplifications (mainly, linearization) are naturally required for such analysis. However, other common assumptions (homogeneous machines modeled by swing equations) appear too narrow from the practitioners' perspective. To bridge the gap between this theory and power engineering practice, we pursue two objectives in this paper.  On one hand, to extend the ``system norm" approach to the more realistic situation of machines with heterogeneous ratings and where the turbine dynamics is accounted for; secondly, to incorporate step-response metrics (Nadir, RoCoF) for an appropriately defined global system frequency, and to separately characterize inter-area oscillations.
Our main contributions are as follows:
\begin{enumerate}
\item We identify a family of heterogeneous machines for which the previous objectives are attainable exactly in closed form. This occurs when dynamic parameters are \emph{proportional} to machine rating, a mild restriction compared with  homogeneity. In this case, we show how to  decompose the system step response into a \emph{system frequency} $\overline{w}(t)$ (motion of the center of inertia) and a transient component of inter-area oscillations.
\item In the above coordinates, we provide natural definitions of performance metrics relevant to power engineers: Nadir and RoCoF are defined as the $\mathcal{L}_\infty$ norms of, respectively, $\overline{w}(t)$ and
$\dot{\overline w}(t)$. A synchronization cost measuring oscillations is defined as the $\Lt$ norm of the vector of deviations. We give a general result on how the latter cost is computed, separating the topological and dynamic components.
\item We apply these methods to different machine models, in particular beyond the traditional swing equation model to include turbine control. The theory helps reveal the impact on relevant parameters; in particular we find the rather limited role of system inertia on synchronization performance,
{\color{black} as compared to damping. } 
\item We analyze the situation where our proportionality assumption does not hold, invoking robust control tools. Although in this case it is not possible to perfectly decouple system-wide and inter-area dynamics, we provide bounds on the deviation from the idealized case. We show in particular how the desired decomposition emerges in the limit of high network connectivity.
\item We provide simulation studies to validate the approach, based on models of the Icelandic power grid \cite{NetData:2018}. We first modify the machine data
to impose paramter proportionality, but respecting the heterogeneous ratings. Our results provide insightful conclusions on the role of the different parameters. Finally we study the real (non-proportional) network, and also explore the influence of connectivity.
\end{enumerate}



The rest of the paper is organized as follows. Section \ref{sec.prelim} contains some background material on the network model. In Section \ref{sec.diag} we introduce our proportionality assumption, carry out the relevant decomposition and introduce the system norms, expressed in terms of a representative machine and the network structure. In Sections \ref{sec.swing} and \ref{sec.turbine} we apply the
results to different machine models, respectively to the
second-order swing dynamics, and a third-order model that
incorporates the turbine control. Section \ref{sec.non-prop} addresses the case where the proportionality assumption is removed. In Section \ref{sec.simulations} we present our simulation studies, and conclusions are presented in Section \ref{sec.concl}. Some of the more extensive derivations are deferred to the appendices. Partial versions of the results in this paper were presented in \cite{pm2017allerton}.

\section{Preliminaries}\label{sec.prelim}


\begin{figure}[htb]
\centering
\includegraphics[width=.75\columnwidth]{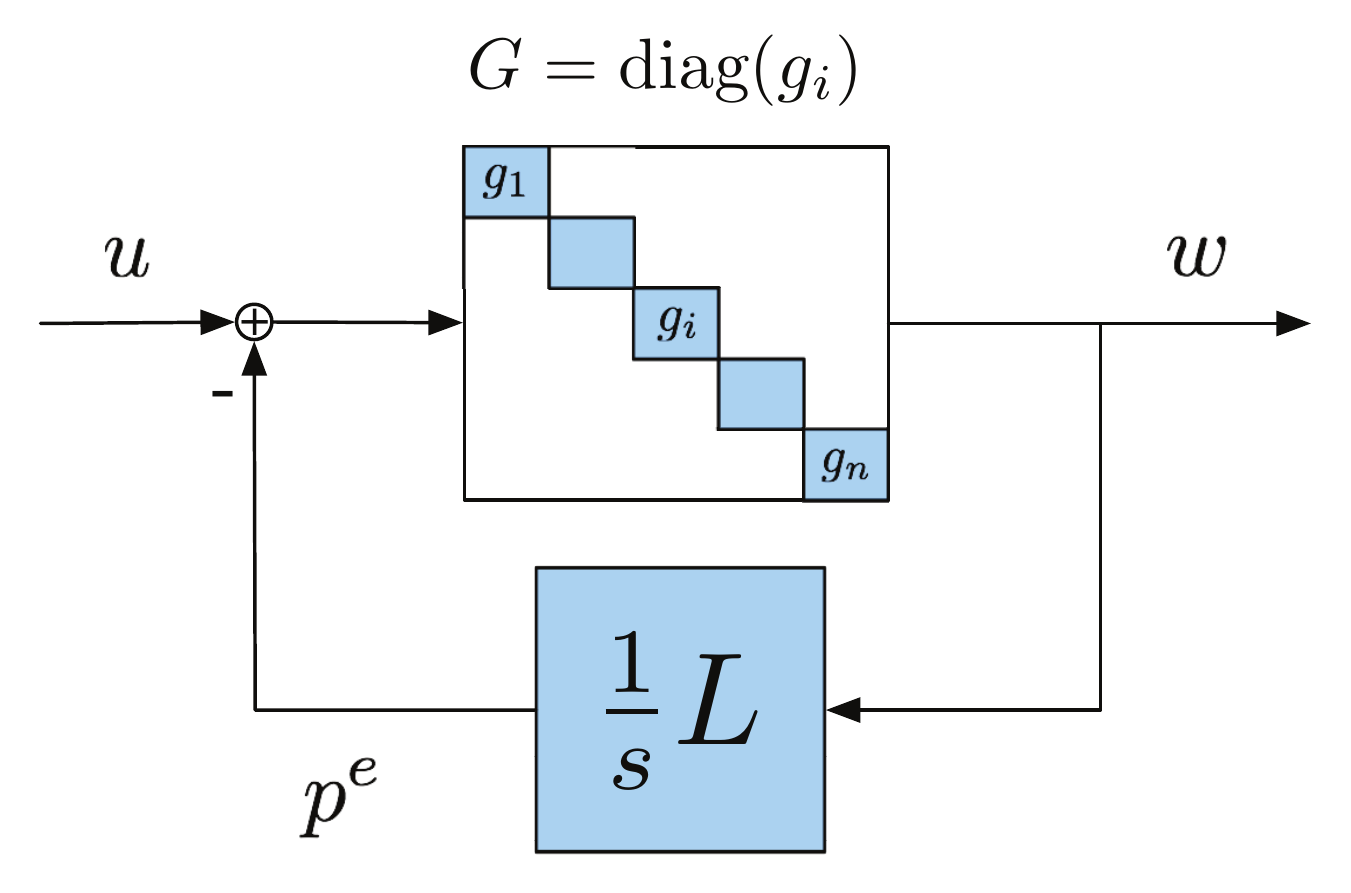}
\caption{Block Diagram of Linearized Power Network}\label{fig.GL}
\end{figure}

We consider a set of $n$ generator buses, indexed by $i\in\{1,\dots,n\}$, dynamically coupled through an AC network. Assuming operation around an equilibrium, the linearized dynamics are represented by the block diagram in Fig. \ref{fig.GL}.

$G(s) = \diag(g_i(s))$ is the diagonal transfer function of generators --with local controllers in closed loop-- at each bus. Each $g_i(s)$ has as output the angular velocity $w_i = \dot\theta_i$, and as
input the net power at its generator axis, relative to its equilibrium value.
This includes an outside disturbance $u_{i}$, reflecting variations in mechanical power or local load, minus the variation $p^e_i$ in electrical power drawn from the network:
\begin{align}
w_i(s) &= g_i(s) (u_i(s)-p^e_i(s)), \quad &i\in \{1,\dots,n\}\label{eq.generator}
\end{align}
The network power fluctuations $p^e$ are given by a linearized model of the power flow equations:
\begin{align}
 p^e(s)&=\frac{1}{s}Lw(s),\label{eq.network}
\end{align}
where $L$ is a undirected weighted Laplacian matrix of the network with elements
\[
[L]_{ij}=\frac{\partial}{\partial{}\theta_j}{\sum_{j=1}^n|V_i||V_j|b_{ij}\sin(\theta_i-\theta_j)}\Bigr|_{\theta=\theta_0}.
\]
Here $\theta_0$ are the equilibrium angles, $|V_i|$ is the (constant) voltage magnitude at bus $i$, and $b_{ij}$ is the line susceptance.

\begin{remark}[Model Assumptions]
The \emph{linearized} network model \reeq{generator}-\reeq{network}
implicitly makes the following assumptions which are standard and well-justified for frequency control on transmission networks \cite{kundur_power_1994}:
\btmz
\item Bus voltage magnitudes $|V_i|$ are constant;
{\color{black}{we are not modeling the dynamics of exciters used for voltage control; these are assumed to operate at a much faster time-scale}}.
\item Lines $ij$ are lossless.
\item Reactive power flows do not affect bus voltage phase angles and frequencies.
\item Without loss of generality, the equilibrium angle difference ($\theta_{0,i}-\theta_{0,j}$) accross each line is less than $\frac{\pi}{2}$.
\etmz
{\color{black}{Also implicit in our model is the so-called Kron reduction, a procedure by which algebraic constraints resulting from buses with no generators are eliminated; see e.g. \cite{bergen_vittal_2000,Ishizaki:je}}.}
{\color{black} For a first principle derivation of the model we refer to \cite[Section VII]{Zhao:2013ts}.} For applications of similar models for frequency control within the control literature, see, e.g., \cite{Zhao:2014bp,Li:2016tcns,mallada2017optimal}.
\end{remark}

Two examples of generator dynamics, to be
considered explicitly in this paper, are:
\begin{example}\label{ex.swing}
The swing equation dynamics
\begin{align*}
m_i \dot{w_i} &= - d_i w_i + u_i - p^e_i,
\end{align*}
{\color{black} which takes the form of the rotational Newton's law, with an inertia $m_i$, and a damping $d_i$; forcing is the power imbalance. The typical swing equations (see e.g. \cite{Ishizaki:je}) are of second order in the machine angles; here we are not making these explicit so the model is of first order.

The corresponding transfer function of Fig. \ref{fig.GL} is}
\begin{align}\label{eq.giswing}
g_i(s) = \frac{1}{m_i s + d_i}.
\end{align}
\end{example}
\mbox{}

\begin{example}\label{ex.turbine}
The swing equation with a first-order model of turbine control:
\begin{align*}
m_i\dot w_i &= - d_i w_i + q_i + u_i - p^e_i,\\
\tau_i \dot q_i &= -r^{-1}_i w_i -q_i.
\end{align*}
Here $q_i$ is the (variation of) turbine power, $\tau_i$ the turbine time constant and $r_i$ the droop coefficient; the governor
dynamics are considered to be faster and neglected. The corresponding transfer function is
\begin{align}\label{eq.giturbine}
g_i(s) = \frac{\tau_i s + 1}{m_i \tau_i s^2 + (m_i+d_i
\tau_i) s+d_i +r_i^{-1}}.
\end{align}
\end{example}
\mbox{}
Of course, other models are possible within this framework. We make the following assumption:
\begin{assumption}\label{ass.strpassive}
The transfer function $g_i(s)$ is stable (has all its poles in $Re(s)<0$) and strictly positive real, i.e. $Re[g_i(j\w)] > 0$ for all $\w \in \R$.
\end{assumption}

This assumption implies that stability of the system, \emph{whichever} the network.

\begin{proposition}\label{eq.prop}
The feedback loop of Figure \ref{fig.GL} is internally stable.
\end{proposition}
\begin{proof}
 The transfer function $\frac{1}{s} L$ is always positive real (see \cite{khalil2002nonlinear}) since the Laplacian is a positive semidefinite matrix. Invoking the passivity theorem (see e.g. \cite{khalil2002nonlinear}) implies the result.
\end{proof}

It is straightforward that the swing model \reeq{giswing} satisfies Assumption \ref{ass.strpassive}. Checking it for the  model \reeq{giturbine} requires  some calculations, or to note that this model is obtained by feedback of \reeq{giswing} with a passive first-order controller.

For more relaxed conditions for internal stability that are robust to network uncertanties we refer the reader to~\cite{pates2018arxiv}.

\section{Modal decomposition for proportionally heterogeneous machines}\label{sec.diag}


A popular research topic in recent years
\cite{bassam,mevsanovic2016comparison,m2016cdc,poolla_dorfler2017,simpson-porco2017,andreasson2017,jpm2017cdc} has
been the application of global metrics from robust control to this
kind of synchronization dynamics, as a tool to shed light on the
role of various parameters, e.g. system inertia. Most of the analytical
results, however, consider a \emph{homogeneous} network where all
machines are identical (i.e., common $m_i$, $d_i$, etc.), a very
restrictive scenario.\footnote{Some bounds on heterogenous systems
are given in \cite{bassam,poolla_dorfler2017}. Numerical studies with heterogeneity are given in \cite{mevsanovic2016comparison}.}

In a real network, where generators have different power
{ratings}, it is natural  for parameters to scale accordingly:
for instance, the inertia $m_i$ of a machine will
grow with its rating, and it is clear that ``heavier" machines
will have a more significant impact in the overall dynamics.

While in principle one would like to cover general parameters, we
will develop most of our theory for parameters satisfying a certain
proportionality.

\begin{assumption}[Proportionality]\label{ass.scale}
There exists a fixed transfer function $g_0(s)$, termed the \emph{representative machine},
and a \emph{rating parameter} $f_i > 0$ for each bus $i$, such that
\[
g_i(s) = \frac{1}{f_i} g_0(s).
\]
\end{assumption}

To interpret this, consider first the swing dynamics of Example
\ref{ex.swing}. Then the assumption is satisfied provided
that inertia and damping are both proportional to $f_i$, i.e. $m_i = f_i  m$, $d_i = f_i d$,
where $m$, $d$ are those of a representative machine. Equivalently, the ratios $m_i/d_i$ are uniform over $i$;
this kind of proportionality is termed ``uniform damping" in \cite{bergen_vittal_2000}.

Going to the case of Example \ref{ex.turbine} with the turbine
dynamics, we find that the assumption is satisfied provided
that $m_i = f_i  m$, $d_i = f_i d$, $r_i^{-1} = f_i  r^{-1}$,
$\tau_i = \tau$; here the inverse droop coefficient is assumed
proportional to rating, but the turbine time-constant is taken to
be homogeneous.

Regarding the practical relevance of our simplifying assumption:
empirical values reported in \cite{oakridge2013} indicate that at least in
regard to orders of magnitude, proportionality is a reasonable
first-cut approximation to heterogeneity, substantially more
realistic than the homogeneous counterpart.
In Section \ref{sec.non-prop} we will discuss deviations from this assumption.

\subsection{Diagonalization}

We will now exploit Assumption \ref{ass.scale} to decompose the dynamics
of Fig. \ref{fig.GL} in a manner that allows for a suitable
decoupling in the analysis. In what follows, $F = \diag(f_i)$
denotes the diagonal matrix of rating parameters. Writing
\[
G(s) = \diag(g_i(s)) = F^{-\ha} [g_0(s) I] F^{-\ha},
\]
we transform the feedback loop into the equivalent form of Fig. \ref{fig.LF}.

\begin{figure}[htb]
\centering
\includegraphics[width=1\columnwidth]{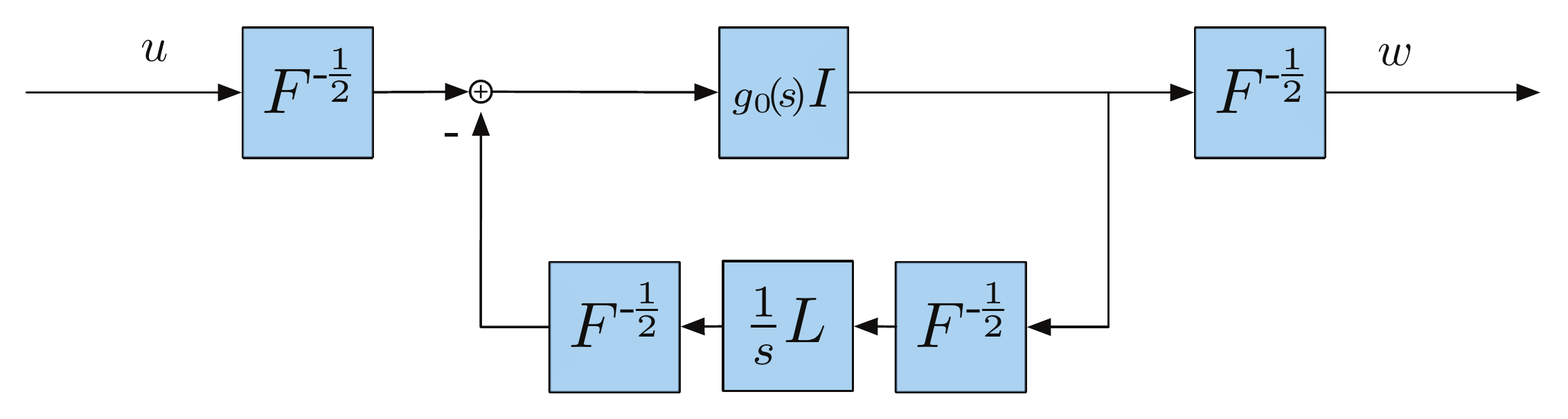}
\caption{Equivalent block diagram for heterogeneously rated machines}\label{fig.LF}
\end{figure}

We introduce a notation for the scaled Laplacian matrix\footnote{This scaling already appears in the classical paper \cite{Winkelman:1981} and, more recently, in~\cite{8264611,guo2018graph}.}
\begin{align}\label{eq.lf}
L_F:=  F^{-\ha} L F^{-\ha},
\end{align}
which is positive semidefinite and of rank $n-1$. Applying the spectral theorem we diagonalize it as
\begin{align}\label{eq.lfdiag}
L_F= V \Lambda V^T,
\end{align}
where $\Lambda = \diag(\lambda_k), \quad 0 = \lambda_0 <
\lambda_1\leq \cdots \leq \lambda_{n-1},$  and $V$ is unitary.
Distinguishing the eigenvector $v_0$ that corresponds to the zero
eigenvalue, we write 
\begin{align}\label{eq.vperp}
V &=\left[v_0 \ \  V_\perp\right], \mbox{ where }\\
& V_\perp \in \R^{n\times (n-1)}, \ V_\perp^T V_\perp = I_{n-1}, \
V_\perp^T v_0 = 0. \nonumber\end{align}
In fact $v_0$ can be made explicit by
recalling that $\ker(L) = \mathrm{span}\{\mathbf{1}\}$, so $\ker(L_F) =
\mathrm{span}\{F^{\ha} \mathbf{1}\}$, from where
\begin{align}\label{eq.v0}
v_0 = \alpha_F  F^\ha  \mathbf{1}, \quad \mbox{ with } \alpha_F:= \Big(\sum_{i} f_i \Big)^{-\ha}.
\end{align}
Substitution of \reeq{lfdiag} into Fig. \ref{fig.LF} and some
block manipulations leads to the equivalent representation of
Fig. \ref{fig.penult}.

\begin{figure}[hbt]
\centering
\includegraphics[width=1\columnwidth]{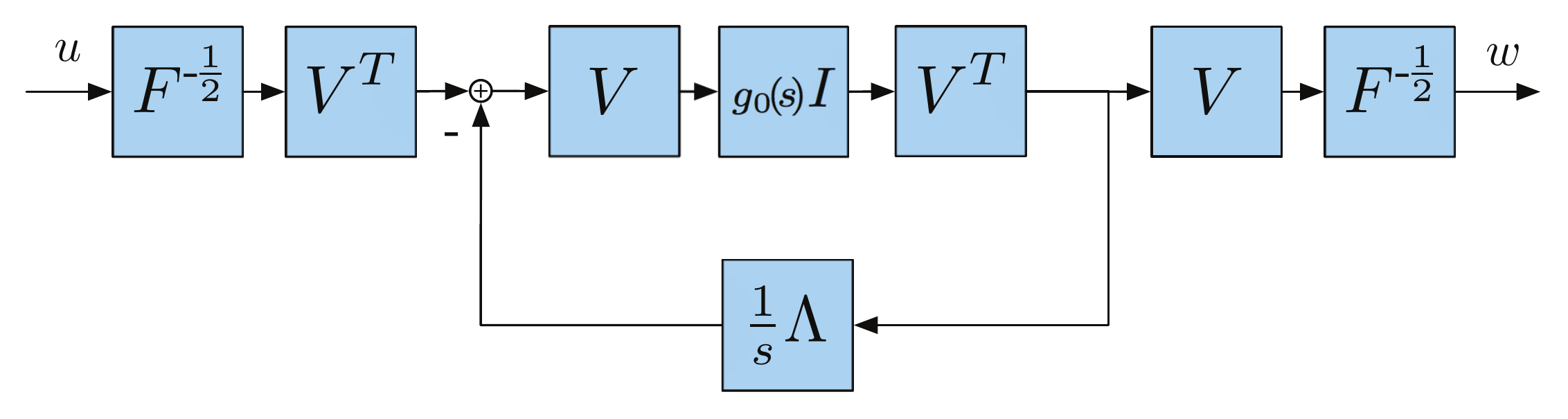}
\caption{Equivalent block diagram for heterogeneously rated machines with diagonalized closed loop}\label{fig.penult}
\end{figure}

Noting finally that the $V$ block commutes with $g_0(s)I$ and thus
cancels out with $V^T$, the internal loop is now fully
diagonalized, yielding the closed-loop transfer function

\begin{align}\label{eq.h}
H(s) &= \diag(h_k(s)), \quad \mbox{with} \nonumber\\
h_k(s) &= \frac{g_0(s)}{1 + \frac{\lambda_k}{s} g_0(s)} = \frac{s g_0(s)}{s + \lambda_k g_0(s)}, \quad
 k=0,1,\ldots,n-1.
\end{align}

Note that by virtue of Assumption \ref{ass.strpassive}, $H(s)$ is a stable transfer function, as a
consequence of the passivity of the integrator feedback.
The overall transfer function between the vector of external power disturbances and the
machine frequency outputs is
\begin{align}\label{eq.twu}
T_{w u}(s) = F^{-\ha} V  H(s) V^T  F^{-\ha}.
\end{align}

\begin{figure}[hbt]
\centering
\includegraphics[width=\columnwidth]{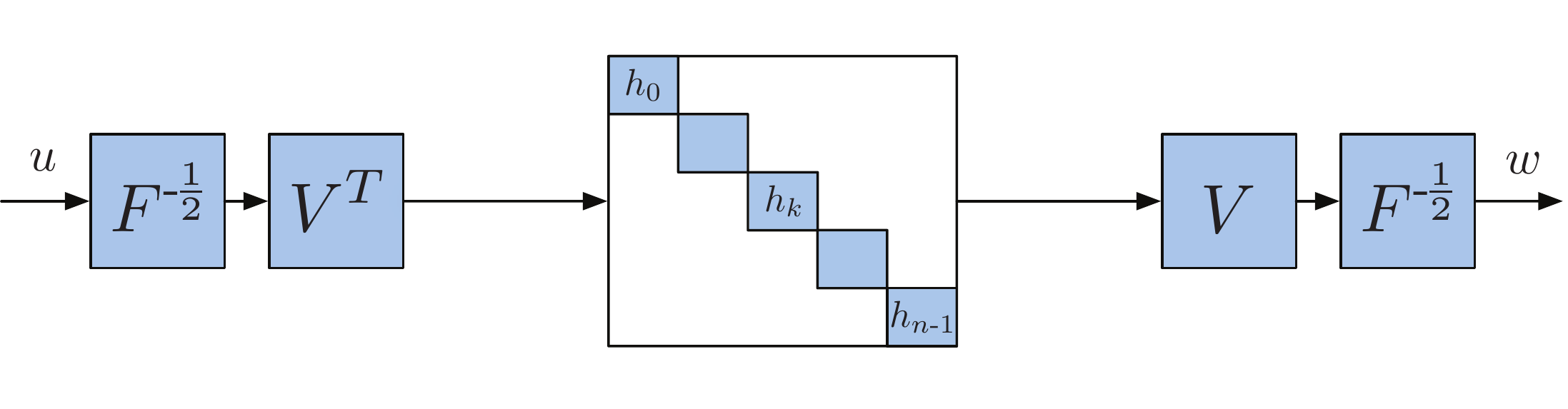}
\caption{Modal decomposition for heterogeneously rated machines with diagonalized closed loop}\label{fig.h}
\end{figure}

\subsection{Step response decomposition}

Global metrics for synchronization performance  in e.g. \cite{bassam,poolla_dorfler2017,simpson-porco2017},
 are system norms ($\Ht$, $\Hinf$) applied to $T_{wu}$ (frequency output), or to a
``phase coherency"  output based on differences in
output angles. The choice of metric carries an implicit assumption
on the power disturbances considered (white noise, or a worst-case
$\Lt$ signal).

In this paper, we wish to analyze metrics which more closely reflect industry practice. {\color{black}{The standard \cite{Anonymous:4nn6T71d} provides a detailed description of the Frequency Response to a \emph{step} change in generation or load, and the associated time-domain performance metrics. The steady-state values of the response are part of the Balancing Authority's requirements; other transient features (Nadir, RoCoF, see below) are also mentioned as important.}}

Now, since buses do not respond in unison, some globally representative \emph{system frequency}
must be considered.   A candidate is the weighted average
\begin{align}\label{eq.coi}
\overline{w}(t) :=  \frac{\sum_i m_i w_i(t)}{\sum_i m_i},
\end{align}
which corresponds to the motion of the \emph{center of inertia} (COI), a classical notion \cite{kundur_power_1994,bergen_vittal_2000}.

We now show that for our family of heterogeneous systems, the
behavior of $\overline{w}(t)$ decouples nicely from the individual bus
deviations $w_i(t)- \overline{w}(t)$, opening the door for a separate
analysis of both aspects of the problem.
{\color{black}
\begin{remark}
We are not considering other disturbances such as line outages which alter the structure of the network itself, in particular the Laplacian matrix. An extreme case would be a system split, which would lead to different COIs in each portion of the disturbed network. These alternatives are a topic for future research.
\end{remark}}

Our input will be a step function $u(t) = u_0 \mathds {1}_{t\geq
0}$; here $u_0$ is a given vector direction. In Laplace
transforms:
\begin{align}\label{eq.stepresp}
w(s) = T_{w u}(s) \frac{1}{s} u_0 = F^{-\ha} V  \frac{H(s)}{s} V^T  F^{-\ha} u_0.
\end{align}
Denote by the diagonal elements of $\frac{H(s)}{s}$ by
\begin{align}\label{eq.htil}
\tilde{h}_k(s) := \frac{h_k(s)}{s} =  \frac{g_0(s)}{s + \lambda_k g_0(s)}, \quad k=0,\ldots, n-1;
\end{align}
we note they are all stable except for $\tilde{h}_0(s)=g_0(s)/s$. This suggests the following decomposition:
\begin{align*}
V  \frac{H(s)}{s} V^T  =
\tilde{h}_0(s)v_0 v_0^T + V_\perp\tilde{H}(s) V_\perp^T,
\end{align*}
where $\tilde{H}(s) =\diag\left(\tilde{h}_1(s),\ldots,\tilde{h}_{n-1}(s)\right)$,
and we have invoked the submatrices of $V$ from \reeq{vperp}.
Furthermore, recalling
from \reeq{v0} that $ v_0 = \alpha_F F^{\ha} \mathbf{1}$, we have
\[
F^{-\ha} v_0 v_0^T F^{-\ha} u_0 = \alpha_F^2 \mathbf{1}\mathbf{1}^T u_0 = \frac{\sum_i u_{0i}}{\sum_i f_i}\mathbf{1}.
\]
Substitution into \reeq{stepresp} gives the decomposition
\begin{align}\label{eq.wbar}
w(s) : =  \overbrace{\frac{\sum_i u_{0i}}{\sum_i f_i}\tilde{h}_0(s)}^{\mbox{\normalsize $\overline{w}(s)$}} \cdot \mathbf{1} +\overbrace{F^{-\ha} V_\perp  \tilde{H}(s) V_\perp^T  F^{-\ha} u_0}^{\mbox{\normalsize $\tilde{w}(s)$}};
\end{align}
or in the time domain,
\begin{align}\label{eq.wdecomp}
w(t) = \overline{w}(t) \mathbf{1}+\tilde{w}(t),
\end{align}
interpreted as follows:
\btmz
\item $\overline{w}(t)$ is a \emph{system frequency} term, applied to all nodes;
\item the transient term $\tilde{w}(t)$ represents the individual node deviations from the synchronous response.
\etmz

\subsection{System frequency}

We can obtain more information on the system frequency by observing that since
$\mathbf{1}^T F^\ha V_\perp  = \alpha_F^{-1} v_0^T V_\perp = 0$, we have
\[
(\mathbf{1}^T F)\tilde{w}(s) \equiv 0.
\]
Therefore $ \mathbf{1}^T F w(t) =  \overline{w}(t) \mathbf{1}^T F
\mathbf{1}$ by \reeq{wdecomp}, which gives
\[
\overline{w}(t) = \frac{\sum_i f_i w_i(t)}{\sum_i f_i};
\]
the system frequency is a weighted mean of bus frequencies, in
proportion to their rating. Noting that $m_i = m f_i$, it follows
that $\overline{w}(t)$ is exactly the COI frequency from \reeq{coi}.

Also, returning to \reeq{wbar} we have
\begin{align}\label{eq.wbart}
\overline{w}(t)  = \frac{\sum_i u_{0i}}{\sum_i f_i} \cdot \tilde{h}_0(t).
\end{align}
Since $\tilde{h}_0(s)=g_0(s)/s$, then $\tilde{h}_0(t)$ is the step response of the
representative machine. Thus $\overline{w}(t)$ corresponds to the angular frequency observed when exciting the
representative machine (in open loop) with the total system
disturbance normalized by the total scale.

\begin{remark}\label{rem:reduced model}
Note that this result is \emph{independent of $L$}, i.e. the electrical
network does not affect the time response of the system frequency,
only the machine ratings themselves. Thus, when the network dependent term ($\tilde w$) converges fast to zero,  \eqref{eq.wbart} is a natural candidate for a reduced order model similar to the ones recently considered in~\cite{guggilam2017engineering, apostolopoulou2016balancing}.
\end{remark}


In the following sections we analyze its behavior for the
previously discussed examples.

\subsection{Quantifying the deviation from synchrony}

We now turn our attention to the term $\tilde{w}(t)$ which represents individual bus deviations from a synchronous response. A natural way of quantifying the size of this transient term is through the $\Lt$ norm
\[
\|\tilde{w}\|_2^2 = \int_0^\infty |\tilde{w}(t)|^2 dt.
\]
We now show how this norm can be computed in terms of the parameters of the scaled network Laplacian,
and the impulse response matrix $\tilde{H}(t)=\diag\left(\tilde{h}_1(t),\ldots,\tilde{h}_{n-1}(t)\right)$,
Laplace inverse of $\tilde{H}(s)$, which encapsulates all information on the machine model.

\begin{proposition} \label{prop.wtil}
$\|\tilde{w}\|_2^2 = z_0^T Y z_0$, where: \btmz
\item ${Y} \in \R^{(n-1)\times (n-1)}$ is the matrix with elements
\begin{align}\label{eq.defy}
y_{kl} = \gamma_{kl}\ip{\tilde{h}_k}{\tilde{h}_l} &= \gamma_{kl} \int_0^\infty
\tilde{h}_k(t) \tilde{h}_l(t)dt, \\
\mbox{where  }\Gamma = (\gamma_{kl})&:= V_\perp^T F^{-1} V_\perp;
\label{eq.defgamma}
\end{align}
\item $z_0 := V_\perp^T F^{-\ha}u_0 \in \R^{n-1}$.
\etmz
\end{proposition}
\begin{IEEEproof}
With the introduced notation we have
\[
\tilde{w}(t) =F^{-\ha} V_\perp  \tilde{H}(t) z_0,
\]
therefore $\tilde{w}(t)^T \tilde{w}(t)= z_0^T  \tilde{H}(t)\Gamma
\tilde{H}(t)z_0.$ The matrix in the above quadratic form has
elements $\tilde{h}_k(t) \gamma_{kl} \tilde{h}_l(t)$, therefore integration in
time yields the result.
\end{IEEEproof}

\begin{remark} The metric $\|\tilde{w}\|_2^2$ \emph{does} depend on the
electrical network, through the eigenvalues and eigenvectors of
$L_F$.
\end{remark}

\subsection{Mean synchronization cost for random disturbance
step}

Since the cost discussed above is a function of the disturbance
step $u_0$, it may be useful to find its average over a random
choice of this excitation. Recalling that the components $u_{0i}$
correspond to different buses, it is natural to assume them to be
independent, and thus $ E[u_0 u_0^T] = \Sigma^u$, a diagonal
matrix.

Using the fact that $z_0=V_\perp^T F^{-\frac{1}{2}}u_0$ we get
\[
E[z_0 z_0^T] =  V_\perp^T F^{-1} \Sigma^u V_\perp =:\Sigma^z,
\]
and the expectation for the cost in Proposition \ref{prop.wtil} is
\[
E\left[\|\tilde{w}\|_2^2\right] = E[z_0^T Y z_0] = E[\tr(Y
z_0z_0^T) ] = \tr(Y \Sigma^z).
\]

Therefore, given $\Sigma^u$ we define the mean synchronization cost as
\begin{equation}\label{eq.mean.sync.cost}
  ||\tilde w||_{2,\Sigma^u} := \sqrt{E\left[||\tilde w||_2^2\right]}=\sqrt{\tr(Y\Sigma^z)}
\end{equation}

We look at some special cases: \btmz
\item $\Sigma^u = I$ (uniform disturbances). Then $\Sigma^z = \Gamma$,
and
\[E\left[\|\tilde{w}\|_2^2\right] = \tr(Y\Gamma) = \sum_{k,l} \gamma_{kl}^2 \ip{\tilde{h}_k}{\tilde{h}_l}.\]
\item $\Sigma^u = F$. This means disturbance size follows the square root of the bus rating. Here $\Sigma^z = I$, and
\[E\left[\|\tilde{w}\|_2^2\right] = \tr(Y) = \sum_{k} \gamma_{kk} \|\tilde{h}_k\|_2^2.\]
\item $\Sigma^u = F^2$. This is probably most natural, with disturbances proportional to bus rating. Here $\Sigma^z =  V_\perp^T F V_\perp =
\Gamma^\dagger $ (pseudoinverse); $E\left[\|\tilde{w}\|_2^2\right]
= \tr(Y \Gamma^\dagger)$.
 \etmz

\subsection{The homogeneous case}\label{ssec.homog}

If all machines have the same response $g_0(s)$, setting $F=I$ we
can obtain some simplifications: \btmz
\item $\{\lambda_k\}$ are the eigenvalues of the original
Laplacian $L$.
\item The system frequency is the average $\overline{w}(t)= \frac{1}{n} \sum_i w_i(t)$, and
satisfies
\[
\overline{w}(t)= \frac{1}{n} \left(\sum_i u_{0i}\right) \tilde{h}_0(t).
\]
\item $z_0= V_\perp^T u_0$, and $\Gamma = V_\perp^T  V_\perp = I$. Therefore the matrix $Y$ in
Proposition \ref{prop.wtil} is diagonal, $Y = \diag(\|\tilde{h}_k\|^2)$,
and
\begin{align}\label{eq.wtilhomog}
 \|\tilde{w}\|_2^2 =
\sum_{k=1}^{n-1} (z_{0k})^2 \|\tilde{h}_k\|^2, \end{align} where $z_{0k}=
v_k^T u_0$ is the projection of the excitation vector $u_0$ in the
direction of the $k$-th eigenvector of the Laplacian $L$.
\item The mean synchronization cost for $E[u_0 u_0^T] = I$ (here
all the preceding cases coincide) is
\begin{align}\label{eq.h2norm}
E\left[\|\tilde{w}\|_2^2\right] = \sum_{k=1}^{n-1} \|\tilde{h}_k\|_2^2 =
\|\tilde{H}\|_{\Ht}^2,
\end{align}
the $\Ht$ norm of the transfer function $\tilde{H}(s)$. We recall
that this was obtained by isolating the portion $\tilde{h}_0(s)$
corresponding to the synchronized response (in this case,
projecting onto $\mathbf{1}^\perp$). In this form, the cost
resembles other proposals \cite{bassam,simpson-porco2017}, for the price of synchrony, and \cite{andreasson2017} for the evaluation of the synchronization cost under step changes in homogeneous systems.
\etmz

\section{Application to the swing dynamics}\label{sec.swing}

In this section we assume we are in the situation of Example \ref{ex.swing}, i.e., the
representative machine is
\begin{align}\label{eq.repswing}
g_0(s) = \frac{1}{ms+d}.
\end{align}
and $m_i = f_i m$, $d_i = f_i d$ are the individual bus
parameters. The diagonal closed loop transfer functions in
\reeq{h} are
\begin{align*}
h_k(s) &= \frac{s}{ms^2 + ds + \lambda_k}, \quad
 k=0,1,\ldots,n-1;
\end{align*}
and the corresponding step response elements are
\begin{align}\label{eq.hktilswing}
\tilde{h}_k(s) &= \frac{1}{ms^2 + ds + \lambda_k}, \quad
 k=0,1,\ldots,n-1.
\end{align}

\subsection{System frequency}
Inverting the transform for the case $\lambda_0=0$ we have
\begin{align}
\tilde{h}_0(t) = \frac{1}{d} \left( 1 - e^{-\frac{d}{m} t}\right),  \quad t>0,
\end{align}
and therefore invoking \reeq{wbart} we find that
\begin{align}\label{eq.wbartswing}
\overline{w}(t)   = \frac{\sum_i u_{0i}}{\sum_i d_i}  \left( 1 - e^{-\frac{d}{m} t}\right), \quad t>0.
\end{align}
We remark the following:
\btmz
\item Again, $\overline{w}(t)$ does not depend on the electrical network.
\item The first-order evolution of $\overline{w}(t)$ implies there is no overshoot; system frequency never deviates to a ``Nadir" further from equilibrium than its steady-state value.
\item The asymptotic frequency  is the ratio $w_\infty = \frac{\sum_i u_{0i}}{\sum_i d_i}$ of total disturbance to total damping. \textcolor{black}{As it is well known~\cite{kundur_power_1994},} it does not depend on the inertia $m$, which only affects the time constant in which this asymptote is achieved.
\item The maximum RoCoF (rate-of-change-of frequency) occurs at $t\to 0+$, and is given by
\begin{align}\label{eq.rocofswing}
\frac{d}{m} \frac{\sum_i u_{0i}}{\sum_i d_i} =  \frac{\sum_i u_{0i}}{\sum_i m_i};
\end{align}
here the total inertia appears, which is natural in the response to a step in forcing. RoCoF increases for low inertia,
however it need not have a detrimental impact: system frequency
initially varies quickly but never deviates more than $w_\infty$,
independent of $m$. \etmz

\subsection{Synchronization cost}

The synchronization cost $\|\tilde{w}\|_2$ for this case,
can be computed by particularizing the result in Proposition
\ref{prop.wtil}. The following result is proved in Appendix \ref{app.ip}.

\begin{proposition}\label{prop.ipswing}
Let $\tilde{h}_k(s)$ be given in  \reeq{hktilswing}, and $\tilde{h}_k(t)$ its inverse
transform, for $k=1,\ldots,n-1$. Then:
\begin{align}\label{eq.hipswing}
\ip{\tilde{h}_k}{\tilde{h}_l} =\frac{2 d}{m(\lambda_k - \lambda_l)^2 + 2
(\lambda_k + \lambda_l)d^2}.
\end{align}
\end{proposition}
\mbox{}

It follows that the matrix $Y$ in \reeq{defy} will depend on both
inertia $m$ and damping $d$, so in general both have an impact on
the ``price of synchrony". Note however that inertia only appears in
off-diagonal terms, and the matrix remains bounded as $m\to 0$ or
$m\to \infty$; we thus argue that inertia has limited impact. Let us
look at this issue in more detail.

\subsubsection{Homogeneous case}
In the case of homogeneous machines, we saw above that $\Gamma=I$
and $Y$ is diagonal, so inertia disappears completely: indeed
using \reeq{wtilhomog} we have
\begin{align}\label{eq.cost2homo}
\|\tilde{w}\|_2^2 = \sum_{k=1}^{n-1} \frac{(v_k^T
u_{0})^2}{2d\lambda_k}.
\end{align}

The cost is inversely proportional to damping, and the direction
of the disturbance $u_0$ also matters. Recalling that $v_{k}$ is
the $k$-th Laplacian eigenvector, the worst-case for a given
magnitude $|u_0|$ is when it is aligned to $v_1$, the Fiedler
eigenvector.

If the disturbance direction is chosen randomly as in Section
\ref{ssec.homog}, then \reeq{h2norm} gives
\begin{align}\label{eq.cost2homorandom}
E\left[\|\tilde{w}\|_2^2\right] = \sum_{k} \|\tilde{h}_k\|_2^2 =
\frac{1}{2 d}\sum_{k}\frac{1}{\lambda_k} = \frac{1}{2 d}
\tr(L^\dagger);
\end{align}
again a similar result to those in \cite{bassam} for homogeneous
systems.

\subsubsection{Heterogeneous, high inertia case}
Assume for this discussion that all the $\lambda_k$ are distinct;
then as $m\to \infty$ we have $y_{kl} \to 0$ for $k\neq l$, so $Y$
again becomes diagonal, and the cost has the limiting expression
\begin{align}\label{eq.cost2heterohigh}
\|\tilde{w}\|_2^2  \stackrel{m\to
\infty}{\longrightarrow}\sum_{k=1}^{n-1} \frac{\gamma_{kk}z_{0k}^2
}{2 d \lambda_k}.
\end{align}
So the high inertia behavior is of a similar structure to the
homogeneous case in \reeq{cost2homo}. Comparisons are not
straightforward, though, since the scaling factor $F$ affects
$z_{0k}, \gamma_{kk}$ and $\lambda_k$ in each of the above terms.

\subsubsection{Heterogeneous, low inertia case}
If $m\to 0$, then the limiting $Y$ matrix is not diagonal. The
corresponding limiting cost is
\begin{align}\label{eq.cost2heterolow}
\|\tilde{w}\|_2^2 \stackrel{m\to 0+}{\longrightarrow}
\sum_{k,l=1}^{n-1} \frac{ \gamma_{kl}z_{0k}z_{0l}}{d (\lambda_k +
\lambda_l)}.
\end{align}
Note, however,  that the diagonal terms are the same as in the
high inertia case. This suggests that inertia plays a limited role
in the $\Lt$ price of synchrony, even in the heterogeneous machine
case. The simulations of Section \ref{sec.simulations} are consistent with this observation. \\

{\color{black}{Reviewing all results for the swing equation model, one may object that it is perhaps not expressive enough to capture the behaviors of interest. In particular, the lack of an overshoot in the COI step response departs from practical observations \cite{Anonymous:4nn6T71d}; one of the main missing features is the presence of a governor with delayed response.  To address this objection we consider the more detailed model of the following section.}}

\section{Model with turbine dynamics} \label{sec.turbine}

In this section we use the model of Example \ref{ex.turbine}, where the
representative machine is
\begin{equation}\label{eq:g_0s-turbine}
g_0(s) = \frac{\tau s + 1}{m \tau s^2 + (m+d \tau) s+d
+r^{-1}}. \end{equation}

The transfer functions for the step response in \reeq{htil}
are:
\begin{align}
\tilde{h}_k(s) = &\frac{\tau s+1}{m\tau s^3 + (m+d\tau)s^2
+(d+r^{-1}+\lambda_k\tau)s+\lambda_k},\nonumber \\ &k=0,1,\ldots,n-1. \label{eq.htilkturbine}
\end{align}

It can be checked (e.g. by applying the Routh-Hurwitz criterion) that they
are  stable whenever $\lambda_k>0$.

\subsection{System frequency}

We can again use \reeq{wbart} and \eqref{eq:g_0s-turbine} to compute the system frequency
\begin{align}
\overline{w}(t)  &= \frac{\sum_i u_{0i}}{\sum_i f_i} \cdot \tilde{h}_0(t),
\end{align}
but now finding the inverse transform of $\tilde{h}_0(s)$ is more involved.
Using partial fractions we first express
\begin{align*}
\tilde{h}_0(s) &=\frac{1}{d+r^{-1}}
\left(\frac{1}{s} - \frac{s + \left( \frac{1}{\tau}-\frac{r^{-1}}{m}\right)}{s^2 + \left(\frac{1}{\tau} +\frac{d}{m}\right)s + \frac{d+r^{-1}}{m\tau}}  \right)
\end{align*}

The first term provides the steady-state response, which is
\begin{align*}
w_\infty =\frac{\sum_i u_{0i}}{\sum_i f_i}\frac{1}{d+r^{-1}}
=\frac{\sum_i u_{0i}}{\sum_i (d_i+r_i^{-1})};
\end{align*}
this is analogous to the swing equation case, except than the droop control has been added to the damping;
{\color{black} the regulation standards  \cite{Anonymous:4nn6T71d} apply to this quantity.}
Again, \textcolor{black}{as expected,} inertia plays no role at all in this steady-state deviation.

The transient term is a second-order transfer function, which we
proceed to analyze now.
Its behavior critically depends on whether its poles are real or complex conjugate.
In particular, whenever
\begin{equation}\label{eq:omegad}
 \frac{d+r^{-1}}{m\tau}-\frac{1}{4}\left(\frac{1}{\tau} +\frac{d}{m}\right)^2=:\omega_d^2>0
\end{equation}
the system is under-damped with poles $\eta\pm j\omega_d$, and
\begin{align}
&\tilde{h}_0(t) = \mathcal L^{-1}\left\{\frac{1}{d+r^{-1}}\left(\frac{1}{s}-\frac{s+\gamma}{(s+\eta)^2 + \omega_d^2}\right)\right\}\nonumber\\
&=\frac{1}{d+r^{-1}}\!\left[1\!-\!{e^{-\eta t}}\left(\cos(\omega_d t) \!-\! \frac{(\gamma\!-\!\eta)}{\omega_d}\sin(\omega_d t)\right)\right]
\label{eq:g0-turbine}
\end{align}
where
\begin{equation}\label{eq:eta-gamma}
\eta:= \frac{1}{2}\left(\frac{1}{\tau}+\frac{d}{m}\right)\quad\text{and}\quad \gamma:=\left(\frac{1}{\tau}-\frac{r^{-1}}{m}\right).
\end{equation}

The system frequency time evolution is thus given by
\begin{align}\label{eq:wbar-wturbine}
\bar w(t)\!=w_\infty \left[1\!-\!{e^{-\eta t}}\left(\cos(\omega_d t) \!-\! \frac{(\gamma\!-\!\eta)}{\omega_d}\sin(\omega_d t)\right)\right].
\end{align}

A few observations are in order:
\begin{itemize}
\item Including the turbine model has a nontrivial effect on the system frequency $\bar w(t)$. It is the presence of the turbine dynamics that provides the characteristic under-damped behavior that produces a Nadir.
\item We have only provided here the solution of $\bar w(t)$ for the (practically more relevant) under-damped
case.
\item Interestingly, \eqref{eq:omegad} shows that the system may become over-damped by either increasing $m$, or decreasing
 $m$! However, the behavior is different for each case: in the very high
 inertia case the Nadir disappears; whereas when $m$ goes to zero, there is an overshoot in the overdamped response.
Since in practice this occurs only for very low inertia and already way beyond the acceptable deviation, we are justified in our focus on the under-damped case.
 \end{itemize}
We now proceed to compute the Nadir and RoCoF for this situation. The proofs of the following
propositions are found in Appendix \ref{app.nadirrocof}.

\begin{proposition}[Nadir]\label{prop.nadir}
Given a power system under Assumption \ref{ass.scale} with generators containing first order turbine dynamics ($g_i(s)$ given by~\eqref{eq.giturbine}). Then under the under-damped condition \eqref{eq:omegad}, the Nadir of the system frequency $\overline{w}(t)$ is given by
\begin{equation}\label{eq:nadir}
||\overline{w}||_\infty =\frac{\left|\sum_i u_{0i}\right|}{\sum_i f_i} \frac{1}{d\!+\!r^{-1}}\left(1\!+\!\sqrt{\frac{\tau r^{-1}}{m} }e^{\!-\!\frac{\eta}{\omega_d}\left(\!\phi+\frac{\pi}{2}\!\right)}\right),
\end{equation}
where the phase  $\phi\in(-\frac{\pi}{2},\frac{\pi}{2})$ is
uniquely determined by
\begin{equation}\label{eq:sinphi}
\sin(\phi)=\frac{\left(\frac{1}{\tau
}-\eta\right)}{\sqrt{\omega_d^2 + \left(\frac{1}{\tau}-\eta\right)^2}}=\frac{m-d\tau}{2\sqrt{m\tau r^{-1}}}.\vspace{1ex}
\end{equation}
\end{proposition}

The dependence of \eqref{eq:nadir} on $m$ is not
straightforward, as $\phi$, $\eta$, and $\omega_d$ depend on it.
The next proposition shows that the dependence is as expected by
conventional power engineering wisdom.
\begin{proposition}\label{prop.dNadirdm}
Given a power system under Assumption \ref{ass.scale} with generators containing first order turbine dynamics ($g_i(s)$ given by~\eqref{eq.giturbine}). Then under the under-damped condition \eqref{eq:omegad}, the maximum frequency deviation $||\bar w||_\infty$ is a decreasing function of $m$, i.e., $\frac{\partial}{\partial m}||\overline{w}||_\infty<0$.
\end{proposition}

We now turn to the maximum rate of change of frequency.
\begin{proposition}[RoCoF]\label{prop.rocof}
Given a power system under Assumption \ref{ass.scale} with generators containing first order turbine dynamics ($g_i(s)$ given by~\eqref{eq.giturbine}). Then under the under-damped condition \eqref{eq:omegad}, the RoCoF is given by
\begin{equation}\label{eq:rocof}
||\dot{\overline{w}}||_\infty = \frac{\left|
\sum_iu_{0,i}\right|}{\sum_i f_i}\frac{1}{m}.
\end{equation}
\end{proposition}

The main difference with Proposition \ref{prop.nadir} is that, while not trivial to establish, here
the maximum is always achieved at $t=0+$, exactly as in the
second order case of \reeq{rocofswing}.

{\color{black} Again we find, as expected, that the
 RoCoF decreases with $m$. But this initial slope need not be consequential.}\footnote{\color{black} If there are system protections activated by RoCoF, these may need to be recalibrated to operate with low inertia, see \cite{o2014studying}.}

\subsection{Synchronization cost}
The synchronization cost $\|\tilde{w}\|^2$ can once again be computed
through  Proposition \ref{prop.wtil}, which requires finding the
inner products $\ip{\tilde{h}_k}{\tilde{h}_l}$, in this case for the functions in
\reeq{htilkturbine}.

Since the corresponding expression is in general rather unwieldy (see Appendix \ref{app.ip}), we
will present some simpler cases, beginning with $k=l$; the norm from
Appendix \ref{app.ip} is:
\begin{align}\label{eq.hknormturbine}
\|\tilde{h}_k\|_2^2 = \frac{m+\tau(\lambda_k \tau +
d)}{2\lambda_k\left[m(r^{-1}+d) +\tau d(r^{-1}+\lambda_k \tau +
d)\right]}.
\end{align}

\subsubsection{Homogeneous case}

The above expression suffices to analyze the case of homogeneous machines,
where $\Gamma=I$ and $Y$ is diagonal. We have from \reeq{wtilhomog} that
\begin{align*}
\|\tilde{w}\|_2^2 = \sum_{k=1}^{n-1} (v_k^T u_{0})^2 \|\tilde{h}_k\|_2^2;
\end{align*}
from \reeq{hknormturbine} we see that, in contrast to the second
order machine model, the inertia $m$ does affect the
synchronization cost. A closer look at $\|\tilde{h}_k\|^2$ as a (linear
fractional) function of $m$ shows that it is \emph{decreasing} in
$m \in (0,\infty) $, going from
\[
\|\tilde{h}_k\|^2_2 \stackrel{m\to 0+}{\longrightarrow} \frac{1}{2\lambda_k d } \cdot \frac{\lambda_k \tau + d}{r^{-1}+\lambda_k \tau + d},
\]
to
\[
\|\tilde{h}_k\|^2_2 \stackrel{m\to \infty}{\longrightarrow}  = \frac{1}{2\lambda_k d } \cdot \frac{d}{r^{-1}+d}.
\]
So higher inertia is beneficial here. Recalling that the corresponding cost for the swing dynamics is
$\frac{1}{2\lambda_k d }$, we see that it has been reduced. In the high inertia case,
the main change is the increased damping through the droop coefficient $r^{-1}$.

\subsubsection{Heterogeneous, high inertia case}

As mentioned, the formula for $\ip{\tilde{h}_k}{\tilde{h}_l}$ for $k\neq l$ is quite formidable, but
we can give its approximation in the limit of large $m$:
\begin{align*}
\ip{\tilde{h}_k}{\tilde{h}_l} \stackrel{m\to \infty}{\sim}
\frac{2(d+r^{-1})}{m(\lambda_k-\lambda_l)^2}, \quad k\neq l.
\end{align*}
This assumes $\lambda_k\neq \lambda_l$. So if the eigenvalues of
the scaled Laplacian $L_F$ are distinct, we see that again the
matrix $Y$ becomes diagonal as $m\to \infty$. The limiting cost is
\[
\|\tilde{w}\|_2^2 \stackrel{m\to \infty}{\longrightarrow}
\sum_{k=1}^{n-1} \frac{z_{0k}^2 \gamma_{kk}}{2\lambda_k d } \cdot
\frac{d}{r^{-1}+d}.
\]
This expression amounts to reducing to the cost
\reeq{cost2heterohigh} for the swing dynamics, by the
fraction $\frac{d}{r^{-1}+d}$. So the role of the turbine in a
high inertia system is again mainly  a change in the droop
coefficient.

\subsubsection{Heterogeneous, low inertia case}
In the low inertia limit, we find that $\ip{\tilde{h}_k}{\tilde{h}_l} \stackrel{m\to 0}{\longrightarrow} \frac{N}{D}$, where
\begin{align*}
N =\ &  2d (d+r^{-1})+ \tau(2d+r^{-1})(\lambda_k+\lambda_l) + 2 \lambda_k\lambda_l \tau^2,\\
D =\ & 2d (d+r^{-1})^2(\lambda_k+\lambda_l) + d \tau(2d+r^{-1})(\lambda_k+\lambda_l)^2 \\ & +  2 d \tau \lambda_k\lambda_l [2r^{-1}+\tau (\lambda_k+\lambda_l)].
\end{align*}
So the limiting matrix $Y$ is not diagonal, as in the second order
case; an expression analogous to \reeq{cost2heterolow} can be
written. Comparisons between the two are not straightforward here, and must be pursued by numerical experimentation; some of these are reported in Section \ref{sec.simulations}.

{\color{black}{Despite the more intricate formulas that preclude from simple conclusions, one observation can be made: system inertia still does not play a major role in the synchronization cost. In particular it does not explode in either end of the inertia spectrum, reflecting a certain degree of indifference; we will give further comments after our experimental results.}}

\section{Beyond the proportionality assumption} \label{sec.non-prop}

The analysis of the preceding sections was predicated on the proportionality of Assumption \ref{ass.scale}. Without it, it is not possible to isolate the center of inertia response from the remaining oscillatory fluctuations.  Nevertheless, some conclusions can be drawn by analyzing perturbations from the proportionality assumption. We start with the swing equation case.

\subsection{Non-proportionality in swing dynamics}

We consider the situation of Example \ref{ex.swing}, but where the mass to damping ratios are not uniform across buses.
Choosing the inertia as indication of scale, we will continue to write $m_i = m f_i$ for a representative mass $m$.
Since the ratios $d_i/f_i$ are not uniform, we will select as representative damping
\begin{align}\label{eq.nomdamp}
d:=\frac{\sum_i d_i}{\sum_i f_i},
\end{align}
and define for each bus the perturbation parameter
\begin{align*}
\delta_i:=d - \frac{d_i}{f_i}.
\end{align*}
A key property of the chosen definition is that
\begin{align}\label{eq.netfdelta}
\sum_i f_i \delta_i= 0.
\end{align}
We now find an expression for the transfer function $g_i(s)$, in terms of the representative machine
$g_0(s)$ in \reeq{repswing} and the perturbation:
\[
g_i(s) = \frac{g_0(s) f_i^{-1}}{1 - g_0(s) \delta_i}.
\]
This can be verified through standard calculations, and is equivalent to the feedback interconnection of
Fig. \ref{fig.np1}.

\begin{figure}[htb]
\centering
\includegraphics[width=.75\columnwidth]{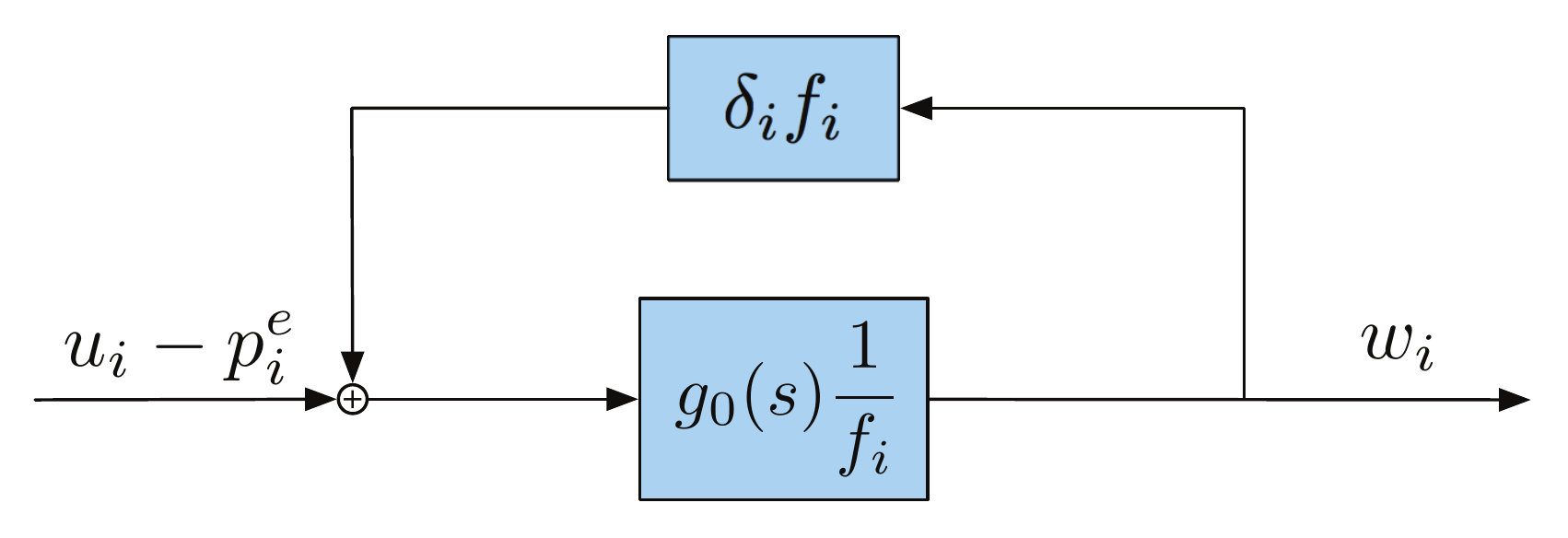}
\caption{Feedback perturbation model for  $g_i(s)$.}\label{fig.np1}
\end{figure}

Including such feedback uncertainty in each machine $g_i(s)$ in Fig. \ref{fig.GL}, and commuting it
with the network feedback, leads to the diagram of 
Fig \ref{fig.np2}. Here $\Delta=\diag(\delta_i)$ is
a diagonal perturbation.

\begin{figure}[htb]
\centering
\includegraphics[width=1\columnwidth]{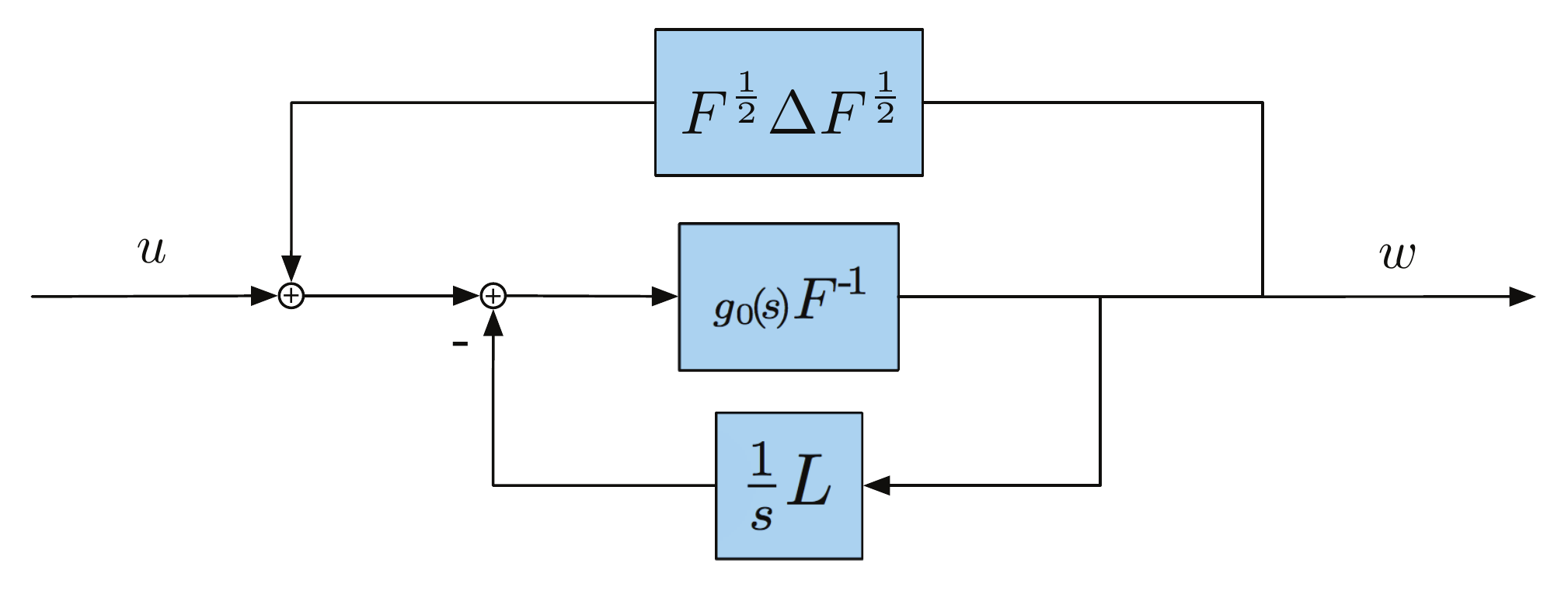}
\caption{Non-proportionality as a diagonal perturbation}\label{fig.np2}
\end{figure}

In this commuted representation, the inner loop corresponds exactly to the proportional situation analyzed in
Section \ref{sec.diag}, yielding the transfer function  $T_{wu}(s)$ in \reeq{twu}. The overall perturbed transfer function is
\begin{align*}
T_{wu}^\Delta(s) = (I - T_{wu}(s)F^\ha \Delta F^\ha)^{-1} T_{wu}(s).
\end{align*}
Replacing with \reeq{twu} leads after some manipulations to
\begin{align}\label{eq.twudelta}
T_{w u}^\Delta(s) = F^{-\ha} V  (I - H(s)\tilde{\Delta})^{-1}H(s) V^T  F^{-\ha},
\end{align}
for $H(s)$ in \reeq{h} and $\tilde{\Delta}: = V^T \Delta V$.

We observe that the above expression reverts back to the one in \reeq{twu} if $\Delta=0$ (nominal case with proportionality). This suggests a robustness analysis for the perturbed performance. A first comment is that robust stability is not an issue: since our non-proportional plant model still satisfies Assumption \ref{ass.strpassive}, the passivity argument for closed loop stability still holds.

Now when considering the step response, the elegant decoupling of Section \ref{sec.diag} breaks down, due to the term
$(I - H(s)\tilde{\Delta})^{-1}$ which does not have diagonal structure, since $\tilde{\Delta}$ is not diagonal. There is no natural way to isolate a ``system frequency"; one may still define it by the motion of the COI, but its dynamics is of high order depending on all synchronization modes.

Some partial analysis is nevertheless possible.

\subsubsection{Steady-state step response}
This quantity can be found through the simple structure of $H(s)$ at $s=0$, namely
\[
H(0) = g_0(0) e_1 e_1^T 
\]
where $e_1$ is the first coordinate vector.
\begin{lemma}\label{lem.dtil}
$H(0)\tilde{\Delta}H(0)=0$, and consequently:
\btmz
\item[(i)] $(I - H(0)\tilde{\Delta})^{-1} = I + H(0)\tilde{\Delta}$,
\item[(ii)] $(I - H(0)\tilde{\Delta})^{-1}H(0) = H(0)$.
\etmz%
\end{lemma}
\begin{proof}
For the first identity it suffices to note that
\[
e_1^T  \tilde{\Delta}e_1 = e_1^T  V^T \Delta V e_1 = v_0^T \Delta v_0 = \alpha_F^2 \sum_i f_i \delta_i = 0,
\]
invoking \reeq{netfdelta}. Then (i) and (ii) follow directly.
\end{proof}
The preceding lemma implies that the nominal and perturbed systems coincide at $s=0$:
\begin{align*}
T^\Delta_{wu}(0) &= T_{wu}(0) = g_0(0) F^{-\ha} V e_1 e_1^T  V^T  F^{-\ha}  \\
 & = \frac{1}{d \sum_i f_i}\mathbf{1} \mathbf{1}^T = \frac{1}{\sum_i d_i}\mathbf{1} \mathbf{1}^T.
\end{align*}
Here we have used that $g_0(0)=1/d$ for the swing model.

Under a step perturbation of direction $u_0$, the residual at $s=0$ of $T^\Delta_{wu}(s)\frac{1}{s} u_0$ is
\[
\frac{1}{\sum_i d_i}\mathbf{1} \mathbf{1}^T u_0 = \frac{\sum_i u_{0i}}{\sum_i d_i}\mathbf{1};
\]
so all buses converge to the same steady-state frequency, with the same value as the one in \reeq{wbartswing}; but here
we cannot claim the same transient behavior for the center of inertia frequency.

\subsubsection{Behavior as connectivity grows}

While one cannot give a clean expression for the transfer function $T^\Delta_{wu}(s)$ outside $s=0$, it is worth considering
what happens when the network connectivity grows, reflected in the growth of the nonzero eigenvalues $\lambda_1\leq \lambda_2 \leq \ldots \leq \lambda_{n-1}$.
Indeed, for any fixed $s$ in the right half-plane $Re(s)\geq 0$, we have
\begin{align}\label{eq.limlam1}
\lim_{\lambda_1 \to \infty}{H(s)} = g_0(s) e_1 e_1^T.
\end{align}
Therefore we can emulate the argument above to conclude that
\begin{align}\label{eq.limlam1twu}
\lim_{\lambda_1 \to \infty}{T_{wu}^\Delta(s)} = g_0(s) \frac{1}{\sum_i f_i} \mathbf{1} \mathbf{1}^T.
\end{align}
This amounts to saying that all buses respond in the same way as the representative machine, appropriately scaled, excited by the aggregate of all bus disturbances. The use of $g_0(s)$ as the natural reduced-order model of the system is thus formally justified (recall Remark~\ref{rem:reduced model}).

The limit can also be extended from a single point $s$ to a compact region of the plane\footnote{Since the transfer functions $h_k(s)$ have a resonant peak at $s=j{\sqrt{\lambda_k}}/{m}$ with value $\frac{1}{d}$, the limit in \reeq{limlam1} cannot be taken uniformly over right-half plane, i.e. convergence in $\Hinf$.}; for instance
one can claim that the Fourier response $T_{wu}^\Delta(j\w)$ for $|\w| \leq B$ converges uniformly as connectivity grows to
a response of the form \reeq{limlam1twu}. So the low bandwidth behavior of the generators becomes highly synchronized, irrespective of non-proportionality.  We will revisit this issue in our simulation section.

\subsection{Non-proportionality in the turbine model}

As mentioned in Section \ref{sec.prelim}, the turbine model can be viewed as the negative feedback interconnection of the swing dynamics
with a controller $K_i(s)$, that adds a power injection as a function of measured bus frequency.

Our development was for the first order controller $K_i(s) = \frac{r_i^{-1}}{\tau_i s + 1}$.  Proportionality in this context required, in addition to the proportional natural damping $d_i = d f_i$, the use of a proportional droop coefficient $r_i^{-1}= r^{-1}f_i $ and a \emph{uniform} $\tau_i = \tau$. In particular in this case we have $K_i(s) = f_i K_0(s)$, where $K_0(s)=\frac{r^{-1}}{\tau s + 1}$ is a fixed transfer function.

Non-proportionality is now modeled as a \emph{dynamic} perturbation of the nominal, normalized controller $K_0(s)$:
\[
\Delta^K_i(s) = K_0(s) - \frac{K_i(s)}{f_i},
\]
Choosing its gain appropriately one can always ensure that
\[
\sum_i f_i \Delta^K_i(s) = K_0(s) \sum_i f_i - \sum_i K_i(s)
\]
is zero at $s=0$. In this way, the steady-state analysis carried out in the swing model extends.

In particular, the feedback structure is analogous to that in Figure \ref{fig.np1}, where the static $\delta_i$ is
replaced by $\delta_i+\Delta^K_i(s)$; with this change, the rest of the analysis can be carried through, to
an expression \reeq{twudelta} where now $\tilde{\Delta}(s) = V^T (\Delta+\Delta^K(s))V$.

Focusing on the behavior at $s=0$, the perturbation will still satisfy the conditions of Lemma \ref{lem.dtil}; therefore the steady-state response is again unchanged from the proportional case. Similar statements can be made about the low-bandwidth behavior as the connectivity grows.

\section{Numerical Illustrations} \label{sec.simulations}

In this section we present numerical simulations to validate our analysis, as well as to provide insight on how generator and network parameters affect the synchronization performance.

Our starting point is a set of real network data available for the Icelandic power grid \cite{NetData:2018}. Since the parameters $m_i$, $d_i$, $r_i$, $\tau_i$ in this data do not satisfy the proportionality of Assumption \ref{ass.scale}, to illustrate the theory we first produce a set of synthetic parameters as follows. First, define $f_i$ and the representative mass using
\begin{equation}\label{eq.fi.m.definitions}
 m:=\frac{1}{n}\sum_i m_i,  \quad f_i := \frac{m_i}{m};
\end{equation}
then define the remaining parameters of the representative machine:  \begin{equation}\label{eq.d.r.tau.definitions}
  d:= \frac{\sum_i d_i}{\sum_i f_i},\quad \textcolor{black}{r^{-1} := \frac{\sum_i r_i^{-1}}{\sum_i f_i}},~
  \text{and}~\tau = \frac{1}{n}\sum_i \tau_i.
\end{equation}
In this way we obtain machines with the same heterogeneity in ratings as the real ones, but satisfying Assumption \ref{ass.scale}; this network is simulated in subsections \ref{ssec.swing proportional} and \ref{ssec.turbine proportional}. Finally in subsection \ref{ssec.non proportional} we will simulate the real network data to assess the effect of the lack of proportionality.

We use the per unit system, standard in the power engineering community, to run our simulations and present our numerical results. We refer to \cite{kundur_power_1994,machowski1997power,Sauer:NpT-MZZG} for more details.

\subsection{Swing Dynamics}\label{ssec.swing proportional}

We first illustrate the behavior of a network obeying the swing dynamics and satisfying the proportionality assumption.

Figure \ref{fig.step.swing.prop} corresponds to the parameters obtained from using \eqref{eq.fi.m.definitions} and \eqref{eq.d.r.tau.definitions} on the Icelandic grid data~\cite{NetData:2018}. It shows the response of the bus frequency $w$, the system frequency $\bar w$ and the synchronization error $\tilde w$, after a disturbance of $-3$p.u. is introduced at bus number 2. A synchronization cost $||\tilde w||_2=4.77$ is incurred.

\begin{figure}[htp]
\centering
\includegraphics[width=\columnwidth]{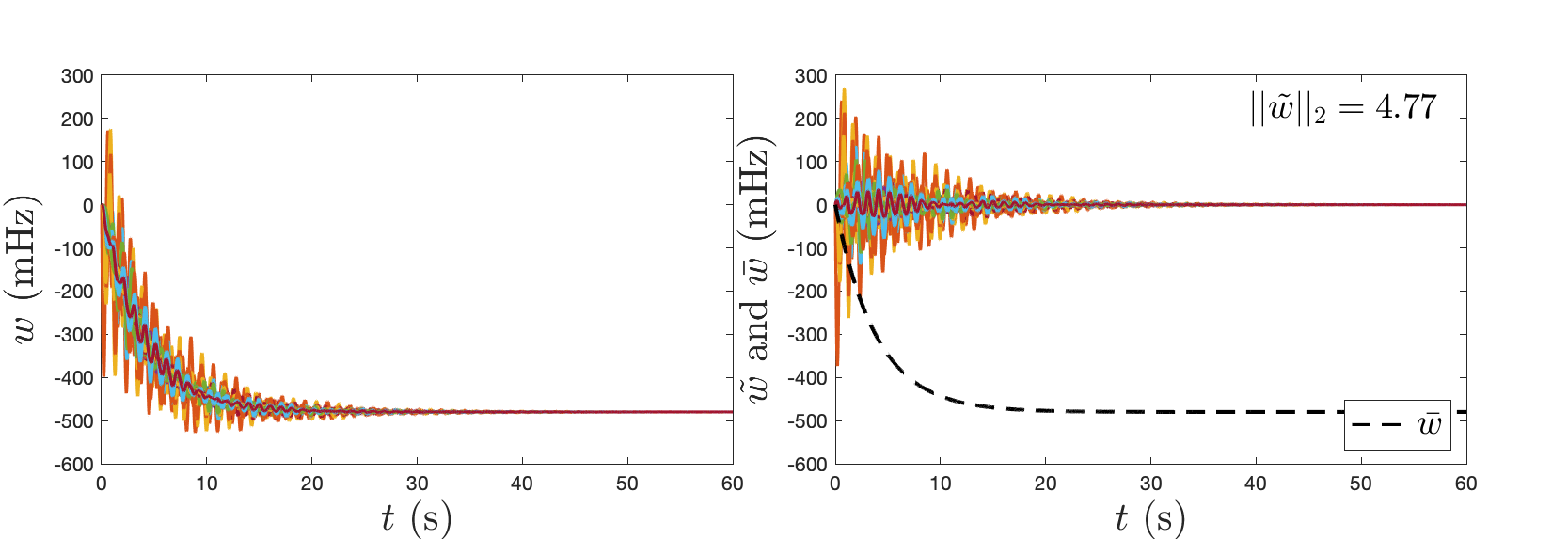}
\caption{Evolution of $w$, $\bar w$ and $\tilde w$ for swing equations with proportional parameters after a step change of $-3$ p.u. at bus 2.}\label{fig.step.swing.prop}
\end{figure}

The system frequency evolution $\overline{w}(t)$ is consistent with a first order response, as predicted by equation \eqref{eq.wbartswing}.  In particular it presents no overshoot, no aggregate oscillatory behavior.

Synchronization performance is reflected in the norm of deviations $\tilde{w}$ from system frequency. To study its dependence with  parameters we use $||\tilde w||_{2,\Sigma_u}$ (averaging in the disturbance direction with $\Sigma_u = F^2$). Figure \ref{fig.norm.vs.m,d} presents values for this metric as a function of $m$ and $d$, with fixed heterogeneity scalings $f_i$.  The red dot represents the parameters used in Figure \ref{fig.step.swing.prop}.  It can be seen that  while changes on the inertia $m$ have nearly no effect on the synchronization cost, changes on the damping $d$ can significantly affect it.

\begin{figure}[htp]
\centering
\includegraphics[width=.75\columnwidth]{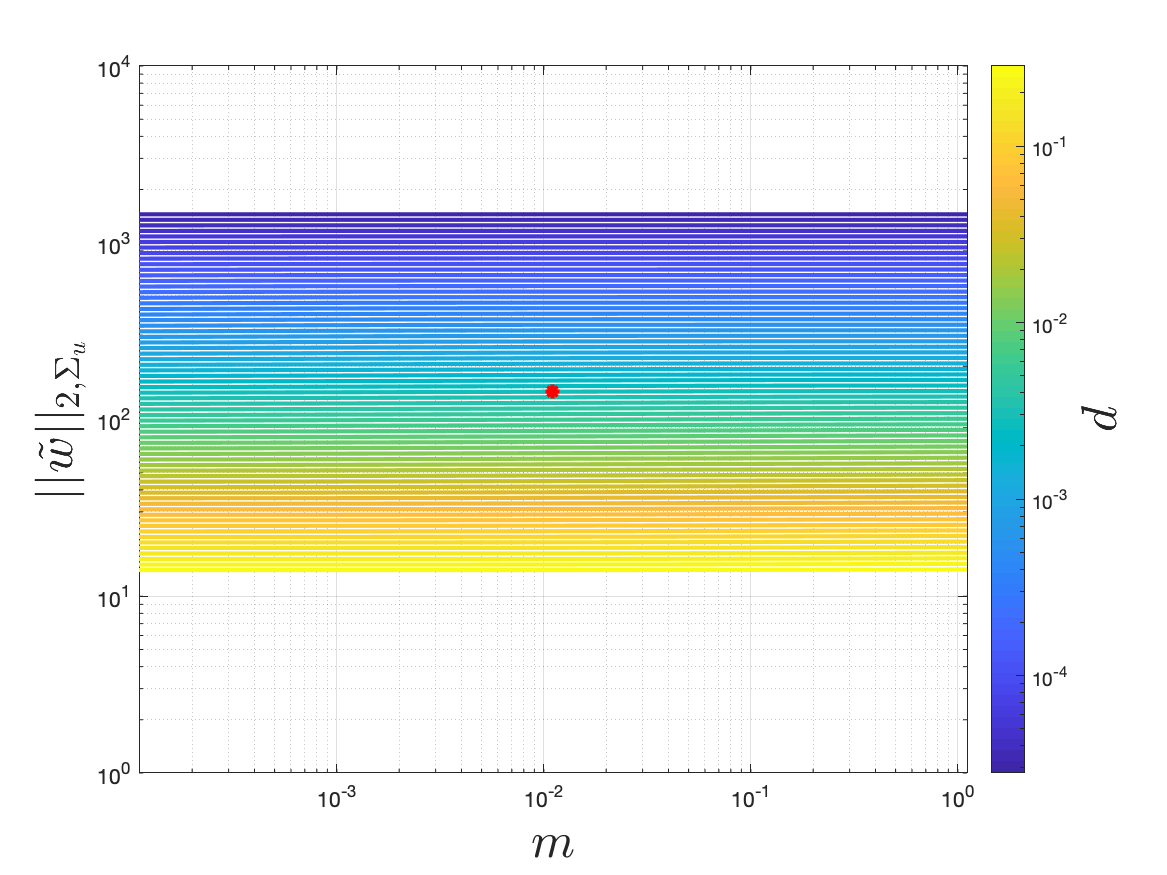}
\caption{Mean synchronization cost $||\tilde w||_{2,\Sigma_u}$ with proportional disturbances ($\Sigma_u=F^2$) as a function of $m$ and $d$.}\label{fig.norm.vs.m,d}
\end{figure}

\begin{figure}[htp]
\centering
\includegraphics[width=\columnwidth]{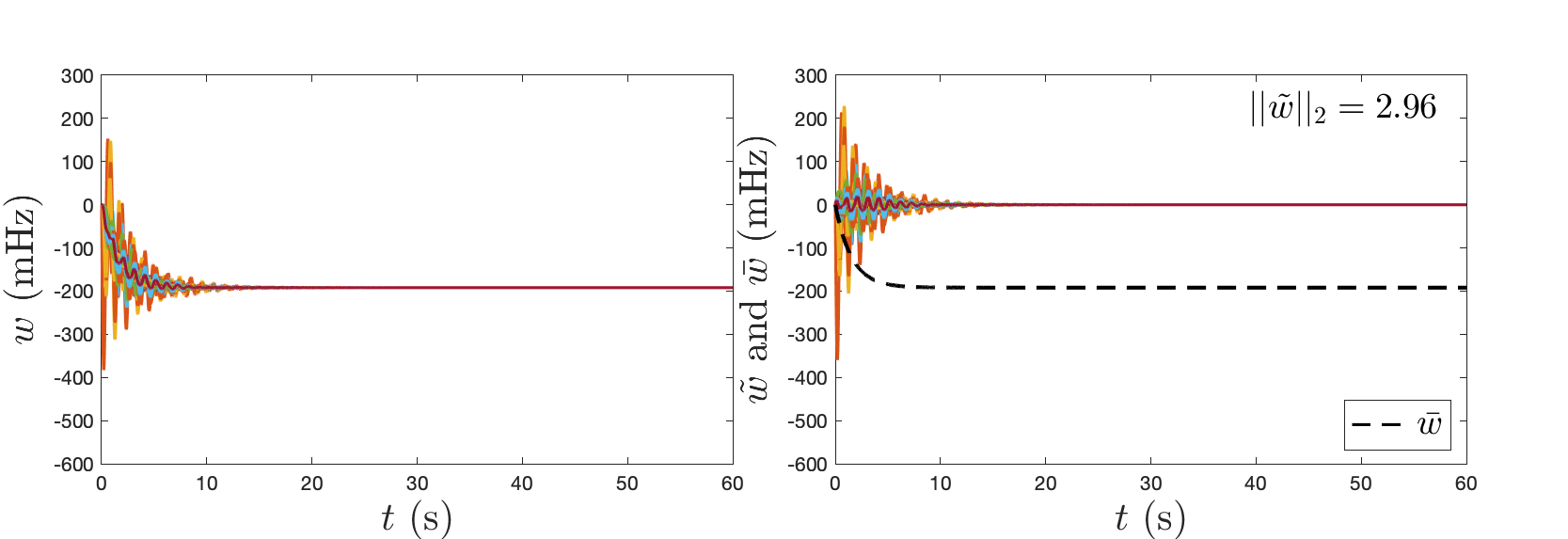}
\includegraphics[width=\columnwidth]{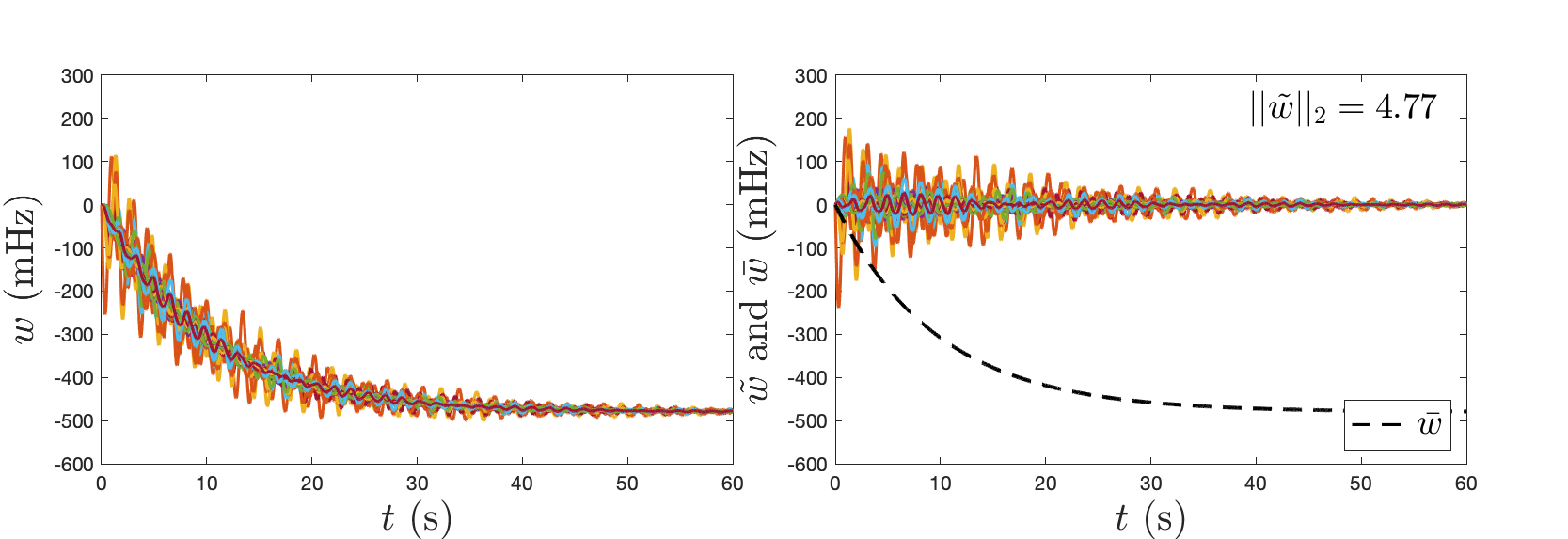}
\caption{Effect of modified parameters in the scenario of Fig. \ref{fig.step.swing.prop}. Top: parameter $d$ is increased 2.5 times. Bottom: parameter $m$ increased 2.5 times. }
\label{fig.step.swing.prop.2.5}
\end{figure}

This conclusion is corroborated by the time simulations in Fig. \ref{fig.step.swing.prop.2.5}, which correspond to increasing by a factor of 2.5 either the system damping $d$ or the inertia $m$. As expected, increasing $d$ not only reduces the synchronization cost
 $||\tilde w||_{2}$ (from 4.77 to 2.96), but also renders a smaller steady-state frequency deviation. Increasing inertia has essentially no effect on our cost, and none on the steady state; it is mainly reflected in a slower rate of convergence.

\subsection{System with Turbine Dynamics}\label{ssec.turbine proportional}

{\color{black}{We now turn to simulations with the model of Section \ref{sec.turbine}, which includes the lagged response of the turbine control.}}

Figure \ref{fig.step.turbine.prop} shows the effect of a $-3$p.u. step change on bus 2. When compared with Figure \ref{fig.step.swing.prop}, the system frequency $\overline{w}$ now is seen to have the characteristic second order response, including an overshoot. The steady-state frequency is closer to the nominal due to the action of the turbine droop ($r^{-1}$).

\begin{figure}[htp]
\centering
\includegraphics[width=\columnwidth]{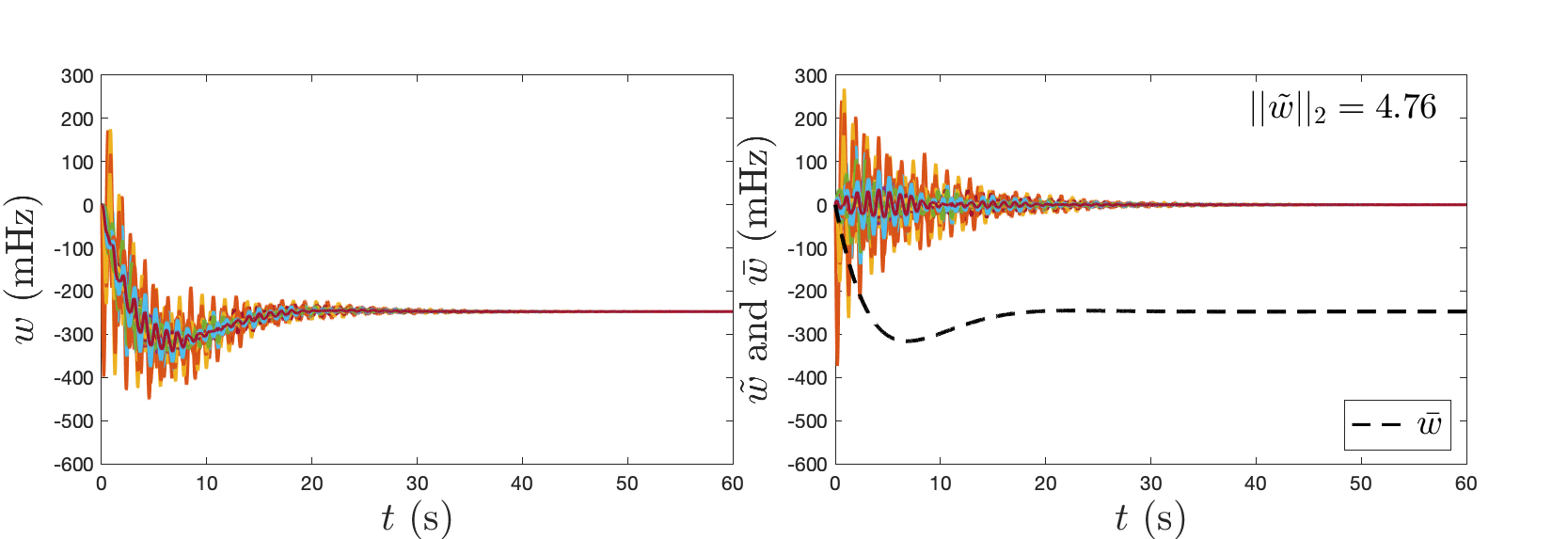}
\caption{Evolution of $w$, $\bar w$ and $\tilde w$ for power grid using swing  and turbine equations with proportional parameters after a step change of $-3$ p.u. at bus 2.}\label{fig.step.turbine.prop}
\end{figure}

{\color{black}{As for the synchronization cost, it is noteworthy that it is essentially unchanged with respect to the swing model. Indeed, if we explore in Fig.   \ref{fig.norm.vs.m,d.turbine} the effect on this cost of varying parameters $m$ and $d$, we find very similar behavior to Fig.  \ref{fig.norm.vs.m,d}.

We supplement this in Fig. \ref{fig.norm.vs.r,tau.turbine}  with a study of the dependence of the mean synchronization cost on parameters $(r,\tau)$; again these are not very significant around the nominal (red) point. In fact we only see an impact for very low values of $\tau$, which would correspond to very fast droop control, essentially equivalent to increasing the damping $d$. This prescription for an improved synchronization is consistent with results in \cite{ulbig2014impact}.
}}

\begin{figure}[htp]
\centering
\includegraphics[width=.75\columnwidth]{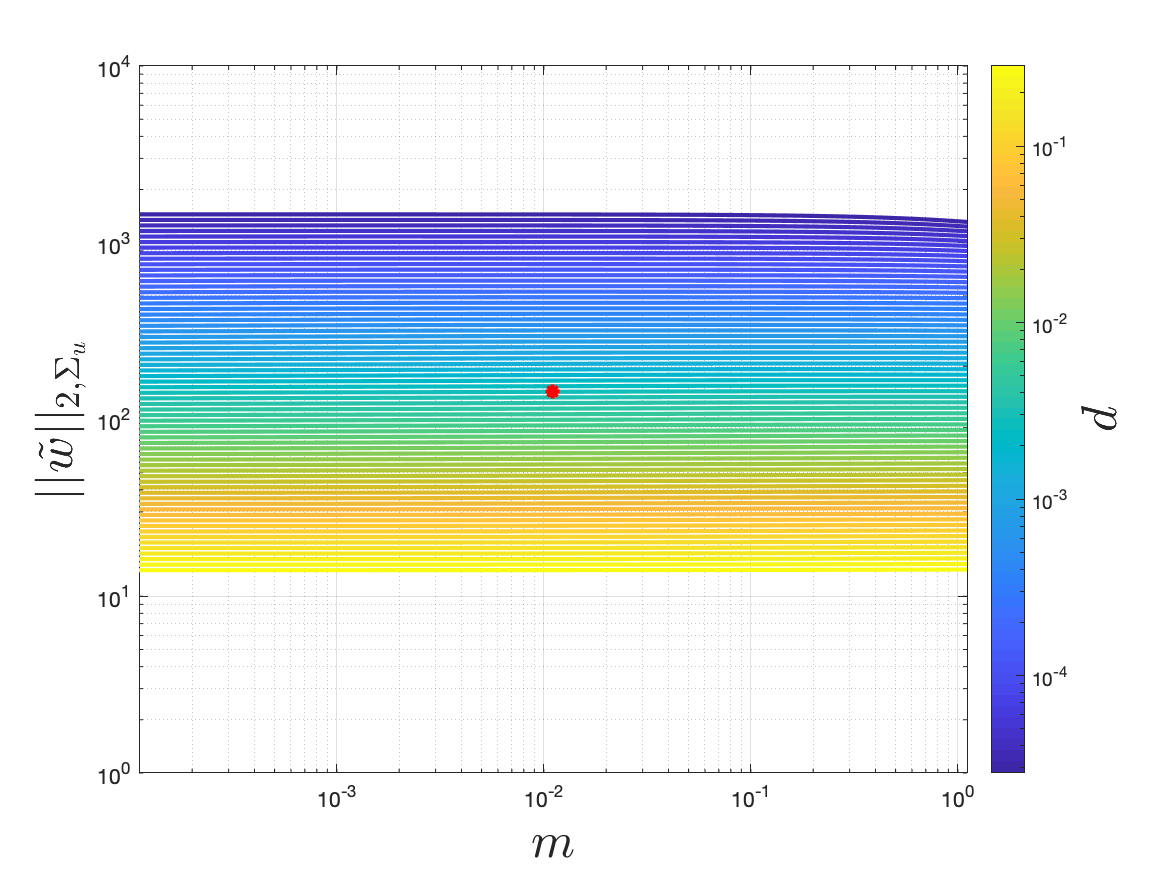}
\caption{Mean synchronization cost $||\tilde w||_{2,\Sigma_u}$ with proportional disturbances ($\Sigma_u=F^2$) as a function of $m$ and $d$.}\label{fig.norm.vs.m,d.turbine}
\end{figure}
\begin{figure}[htp]
\centering
\includegraphics[width=.75\columnwidth]{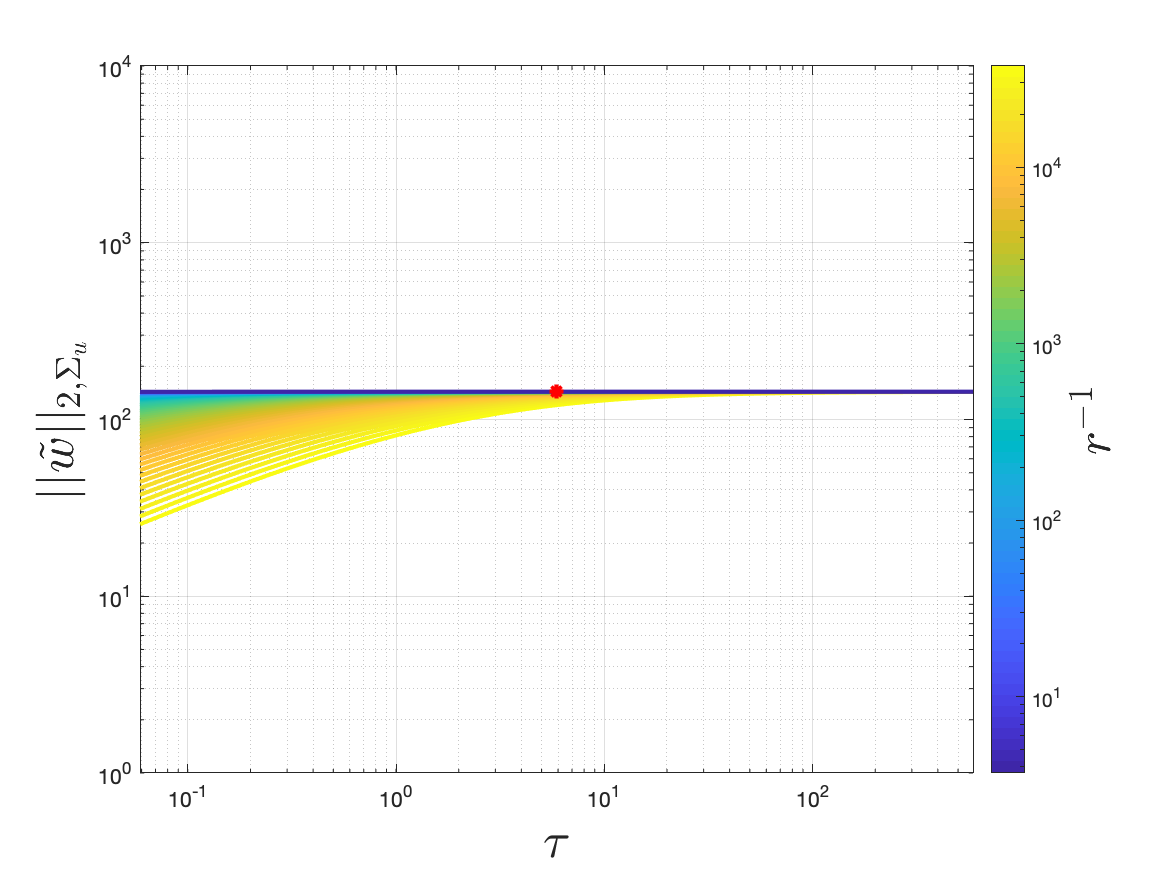}
\caption{Mean synchronization cost $||\tilde w||_{2,\Sigma_u}$ with proportional disturbances ($\Sigma_u=F^2$) as a function of $r$ and $\tau$.}\label{fig.norm.vs.r,tau.turbine}
\end{figure}

\begin{figure}
    \centering
    \includegraphics[width=.8\columnwidth]{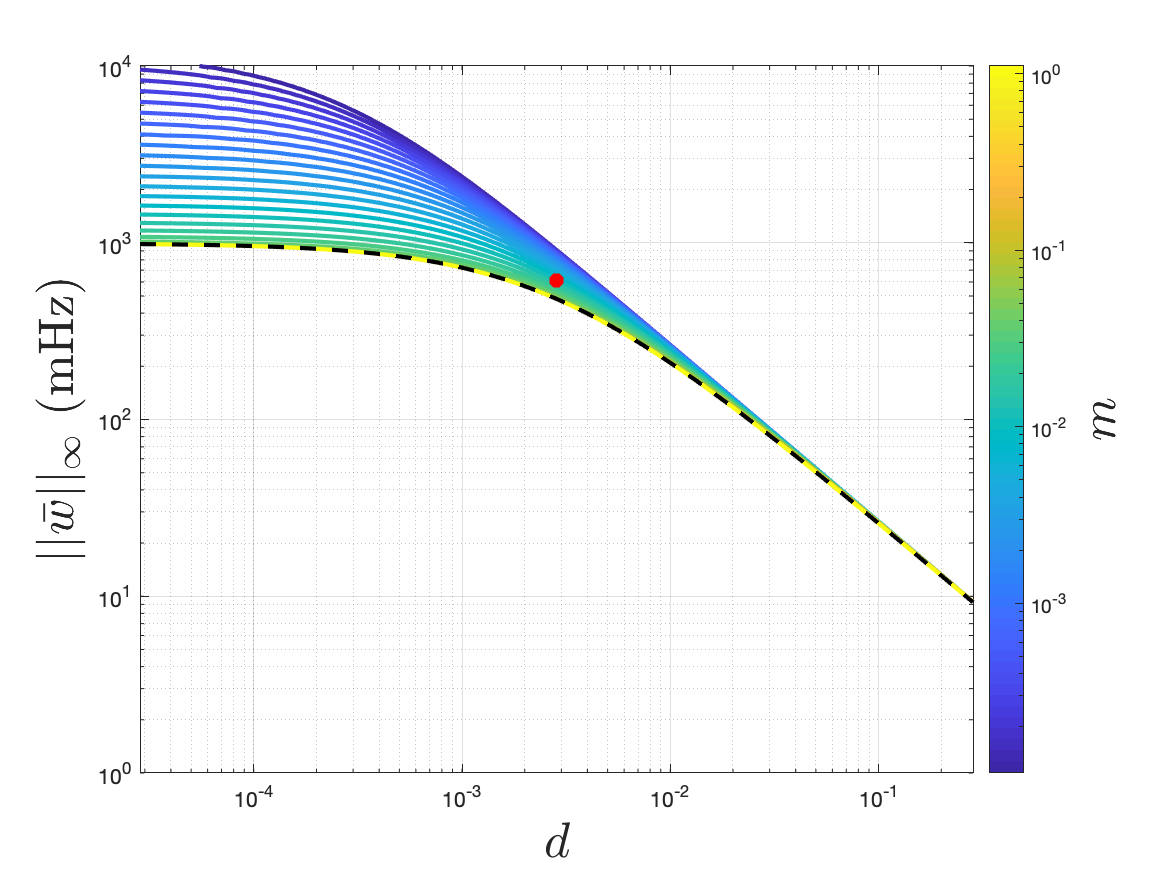}
    \caption{Nadir of the system frequency $\bar w$ ($||\bar w||_\infty$) as a function of $d$ and parametric variation of $m$. The black dashed line represents the steady state frequency $w_\infty$.}
    \label{fig:nadir-vs-d}
\end{figure}
\begin{figure}
    \centering
    \includegraphics[width=.75\columnwidth]{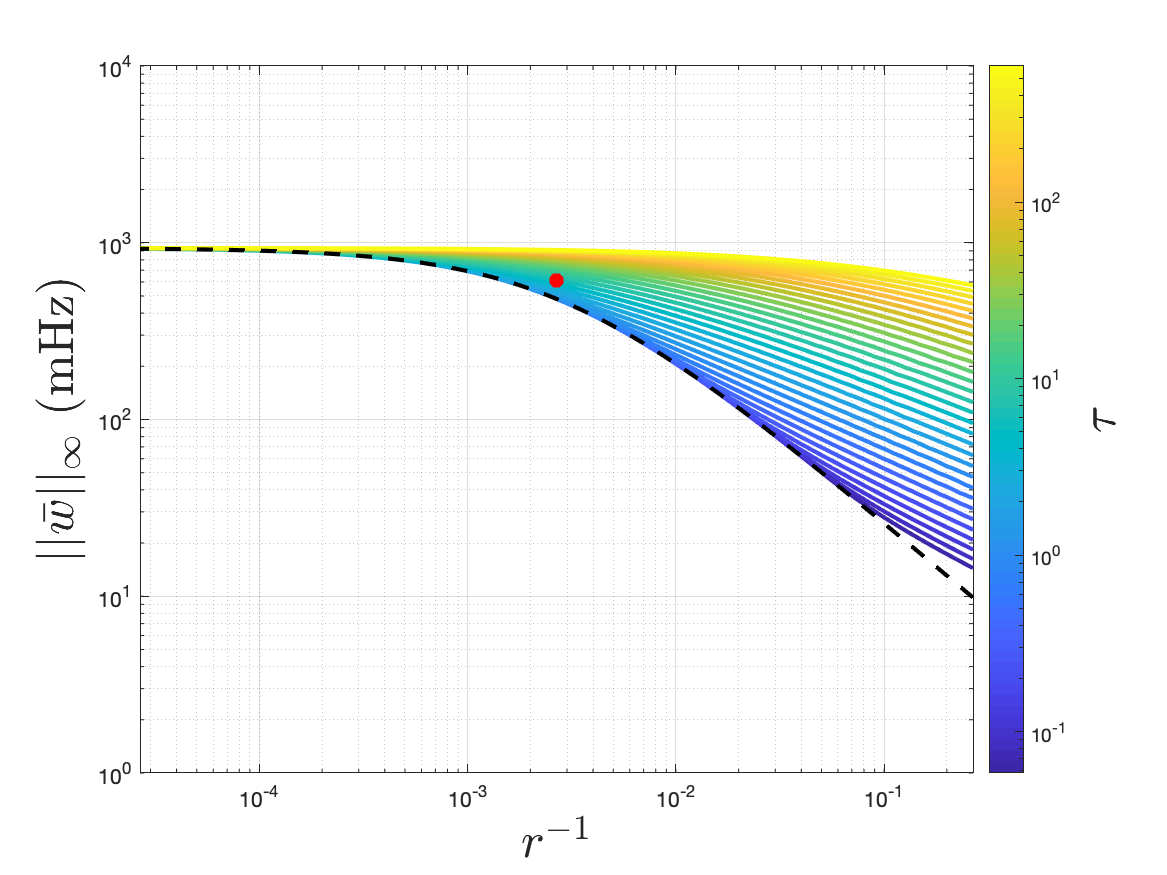}
    \caption{Nadir of the system frequency $\bar w$ ($||\bar w||_\infty$) as a function of $r^{-1}$ and parametric variation of $\tau$. The black dashed line represents the steady state frequency $w_\infty$.}
    \label{fig:nadir-vs-r}
\end{figure}

The dependence of the Nadir $||\bar w||_\infty$ on the parameters is illustrated in Figs. \ref{fig:nadir-vs-d} and \ref{fig:nadir-vs-r}. In
particular, Figure~\ref{fig:nadir-vs-d} shows again the limited impact of an increase of $m$ on the Nadir in comparison to an increase of $d$. Similarly, increasing $r^{-1}$ seems to have a more significant impact than a reduction in $\tau$.
Moreover, while both $d$ and $r^{-1}$ have a direct impact on the final steady-state (which constitutes a bound on the Nadir), neither $m$ nor $\tau$ affect it, as expected.


We further illustrate in Fig. \ref{fig.step.turbine.prop.2.5} the
effect of individually varying parameters $d$, $r$  and $m$ in
the scenario of Fig. \ref{fig.step.turbine.prop}. The most significant impact both in time characteristics of $\overline{w}(t)$
(Nadir, steady state) and on synchronization cost is increasing the damping  $d$.  The droop coefficient $r^{-1}$ has some positive influence in the time response, less so the inertia $m$; neither of the two reduces the synchronization cost.

{\color{black}{
As to the role of inertia, it is illustrative to compare the bottom plot of Fig. \ref{fig.step.turbine.prop.2.5} with Fig.  \ref{fig.step.turbine.prop}. For the system frequency, we see that the increased inertia produces a very slight  reduction in the deviation at the Nadir. The oscillatory components have been reduced in amplitude (the largest almost in half) but now they last longer in time; our $\mathcal{L}_2$ cost is largely insensitive.
Intuitively speaking: a heavier system initially deviates less, but it takes more effort to subsequently control it to settle down.
}}

\begin{figure}[htp]
\centering
\includegraphics[width=\columnwidth]{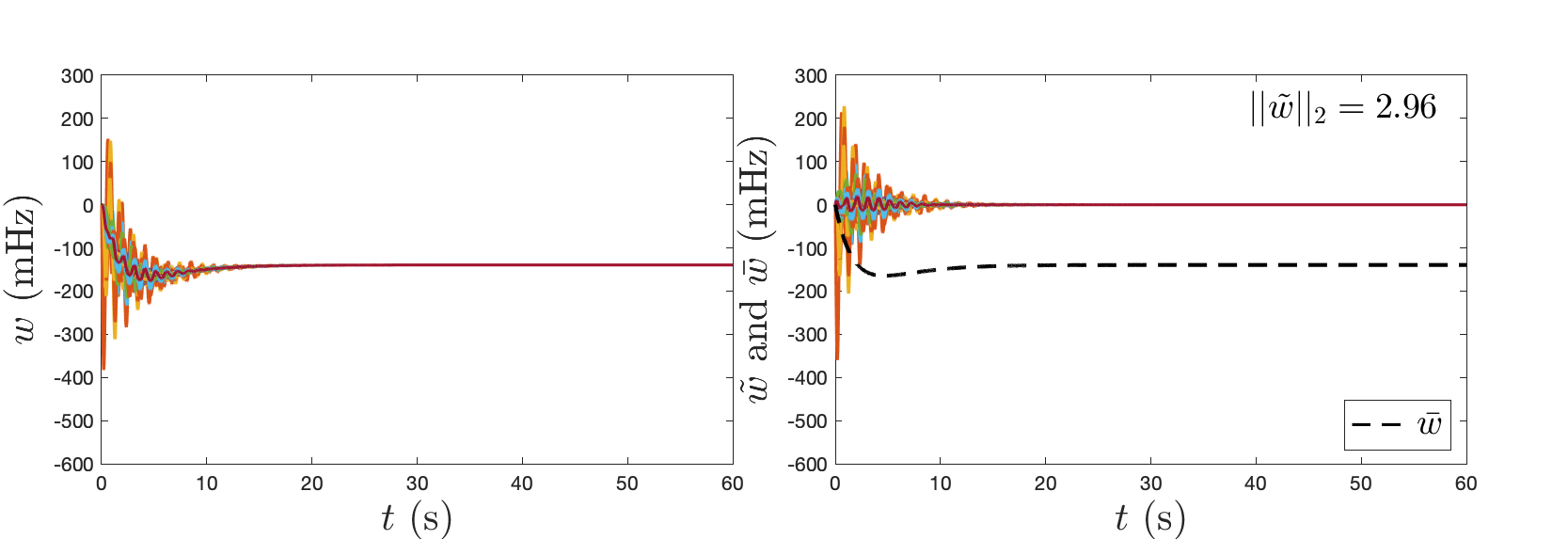}
\includegraphics[width=\columnwidth]{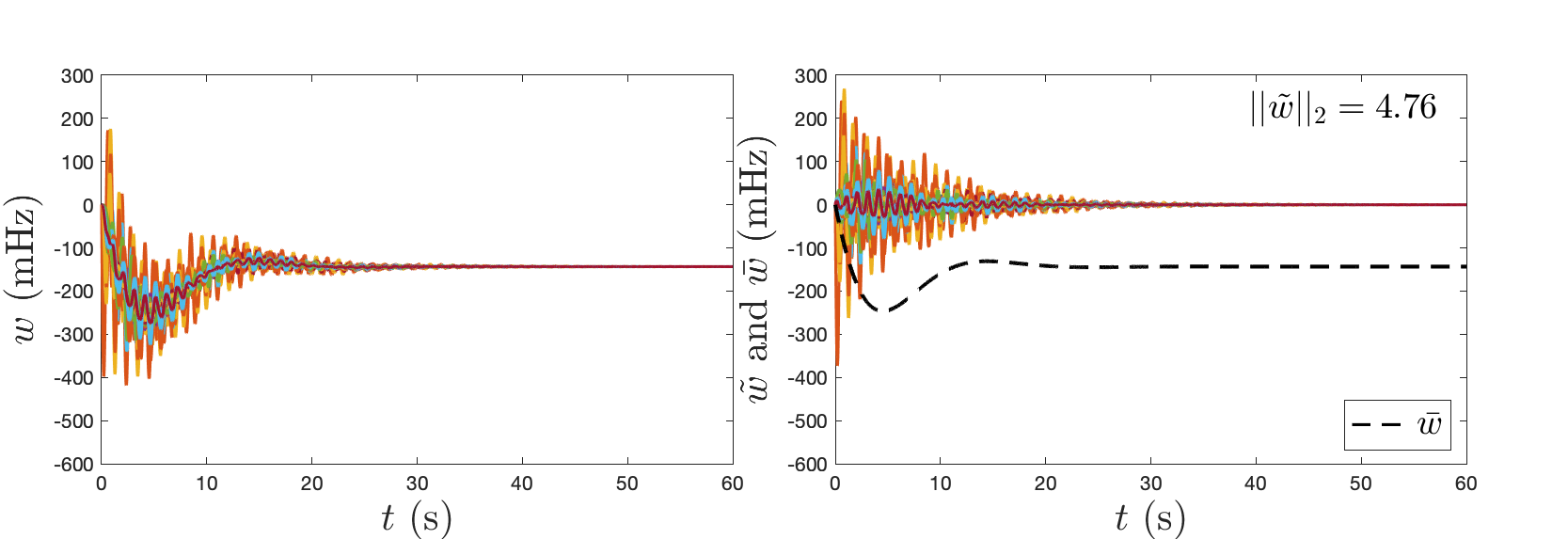}
\includegraphics[width=\columnwidth]{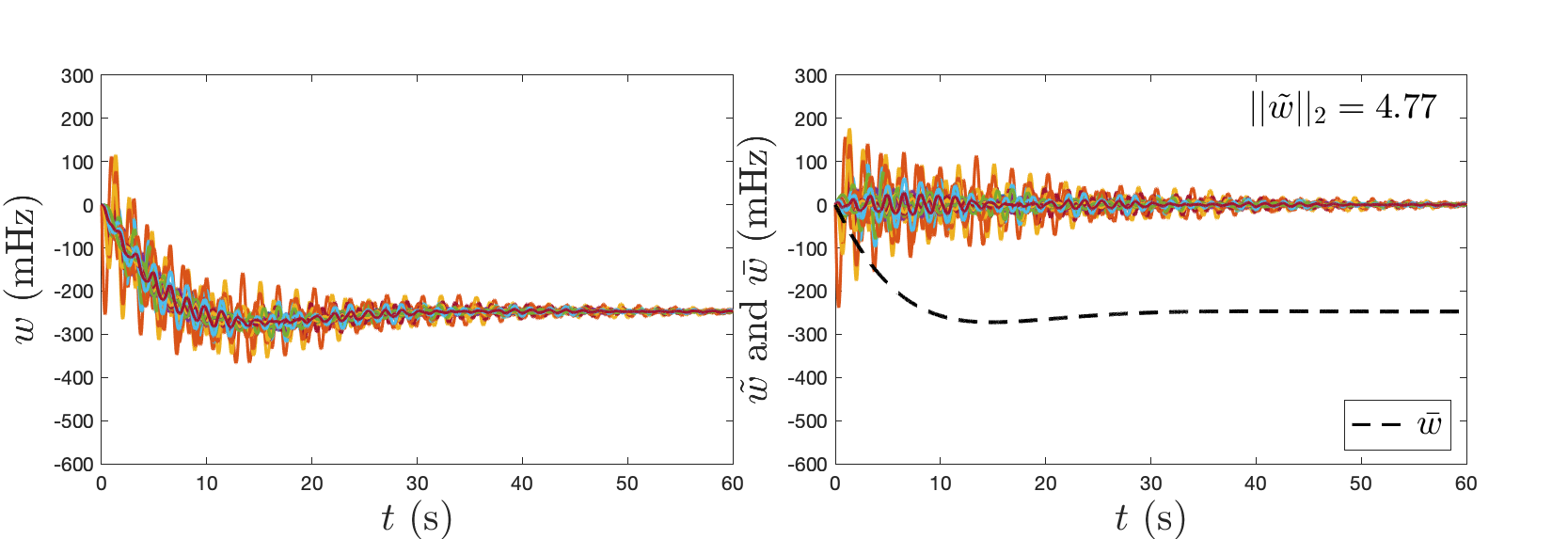}
\caption{Effect of modified parameters in the scenario of Figure \ref{fig.step.turbine.prop}. Top: parameter $d$ increased 2.5 times.  Middle: parameter $r^{-1}$ increased 2.5 times. Bottom: parameter $m$ increased 2.5 times.
}
 \label{fig.step.turbine.prop.2.5}
\end{figure}

\subsection{Non-proportionality  and Impact of Connectivity}\label{ssec.non proportional}

So far we have simulated a network that corresponds to the real-world Icelandic grid, both electrically and in the machine inertias, but where other machine parameters were altered to satisfy Assumption \ref{ass.scale}. What happens if we simulate with the real, unmodified data? Figure \ref{fig.step.turbine.true} represents the response under the same scenario as in Figure \ref{fig.step.turbine.prop}, but using instead the true machine parameters. We observe the close similarity of both sets of plots, indicating that our modeling technique is capturing the essential dynamics. We note that there is a certain reduction in our synchronization cost; real data behave somewhat better than the proportional version.

\begin{figure}[htp]
\includegraphics[width=\columnwidth]{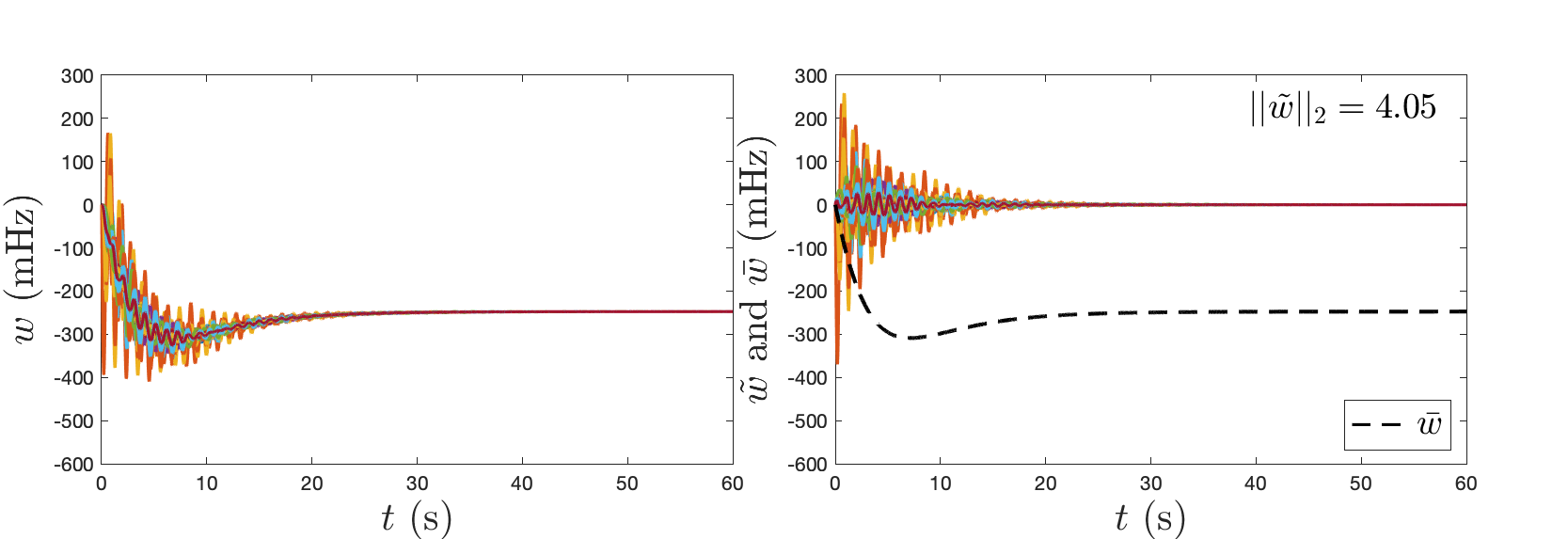}
  \caption{Evolution of $w$, $\bar w$ and $\tilde w$ for the system with turbine dynamics using true parameters obtained from \cite{NetData:2018}. Same disturbance set up as in Figure \ref{fig.step.turbine.prop}\label{fig.step.turbine.true}.}
\end{figure}

Another issue to explore is the effect of  network connectivity; in Section \ref{sec.non-prop} we showed that increased connectivity (reflected in larger eigenvalues for $L_F$) aligns the response around the system frequency, and reduces the impact of uncertainty due to non-proportionality. To investigate this we randomly add lines to the Icelandic power network with impedances uniformly distributed between its minimum and maximum values. Figure \ref{fig.grid} shows some of the resulting network topologies; here $K$ denotes the number of lines added,  $K=0$ in Figure \ref{fig.grid.k0} being the original topology.

\begin{figure}[htp]
\includegraphics[width=\columnwidth]{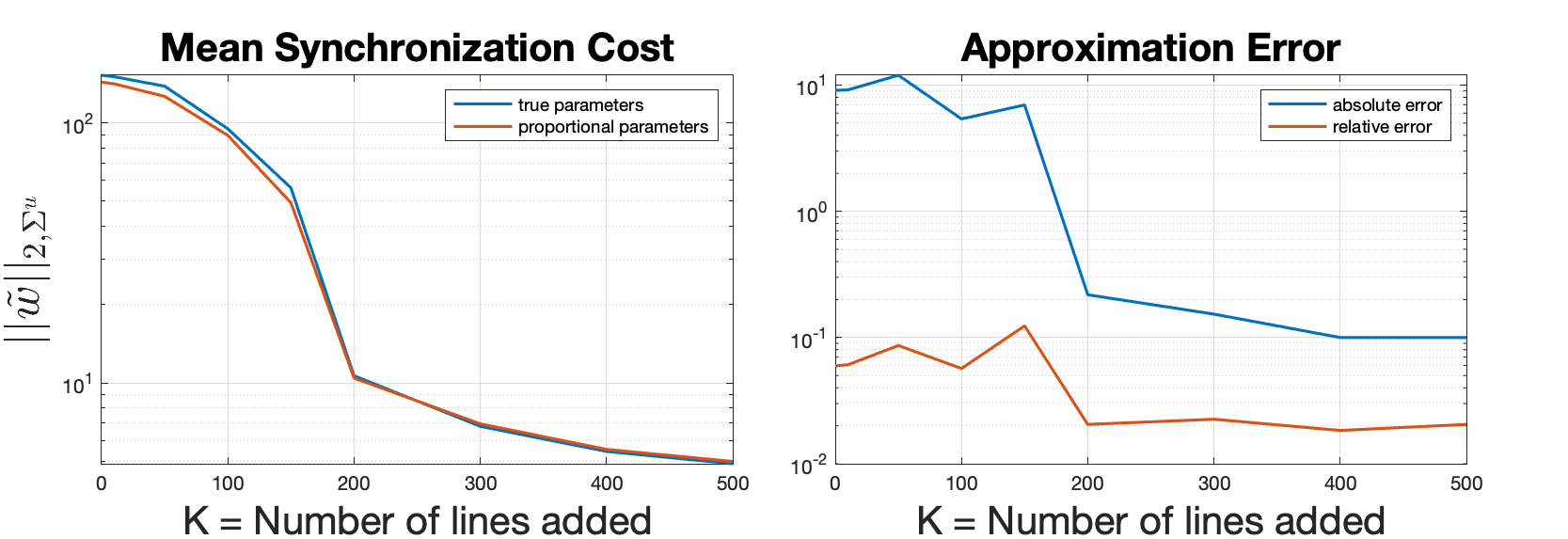}
  \caption{Mismatch error between true parameters obtained from \cite{NetData:2018} and proportional parameters obtained using \eqref{eq.fi.m.definitions} and \eqref{eq.d.r.tau.definitions} of the mean synchronization cost with random disturbances proportional to rating bus ($\Sigma_u=F^2$). }\label{fig.error}
\end{figure}

Figure \ref{fig.error} shows the mean synchronization cost for proportionally correlated disturbances ($\Sigma_u=F^2$) for the true grid parameters as well as the proportional parameters computed using \eqref{eq.fi.m.definitions}-\eqref{eq.d.r.tau.definitions} for the swing dynamics with turbines.
We can see how not only the synchronization error goes to zero as the connectivity increases (as expected) but also the relative error between the synchronization cost with proportional parameters and real parameters remains within 10\%.

\begin{figure*}
\begin{subfigure}{.195\textwidth}
\includegraphics[width=\columnwidth]{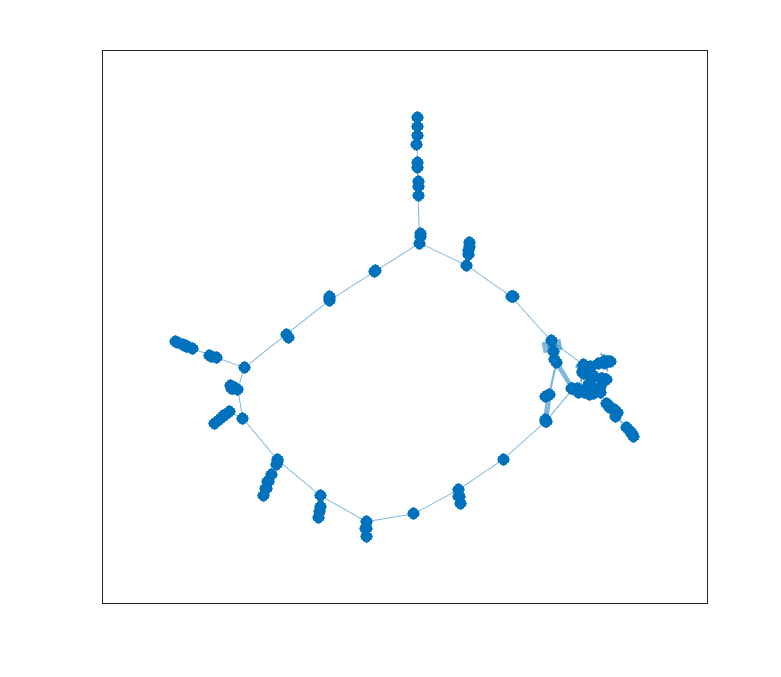}
\caption{$K=0$}\label{fig.grid.k0}
\end{subfigure}
\begin{subfigure}{.195\textwidth}
\includegraphics[width=\columnwidth]{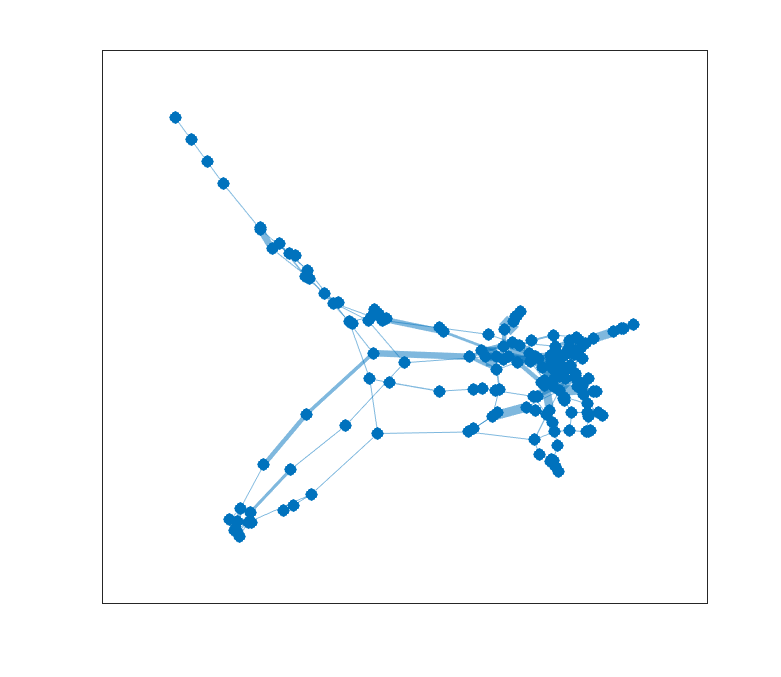}
\caption{$K=25$}\label{fig.grid.k25}
\end{subfigure}
\begin{subfigure}{.195\textwidth}
\includegraphics[width=\columnwidth]{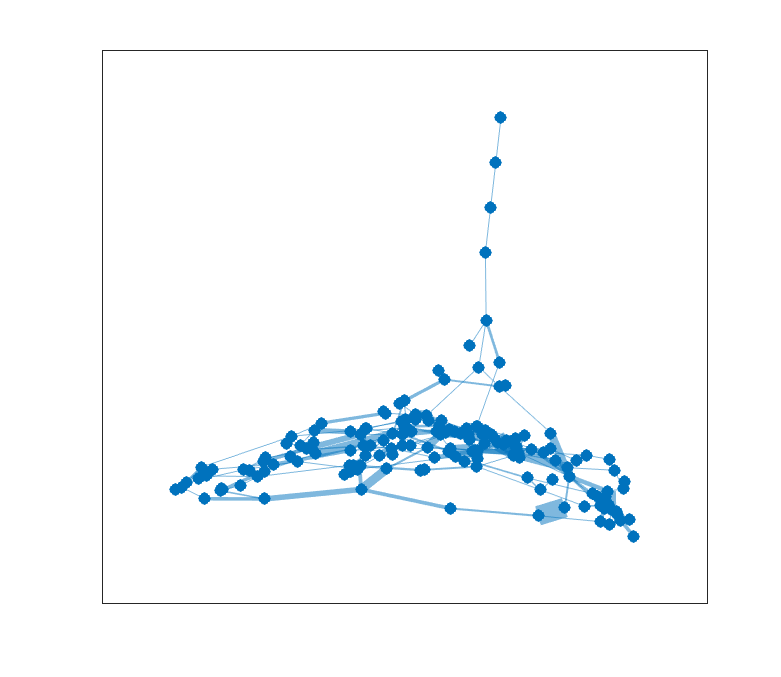}
\caption{$K=50$}\label{fig.grid.k50}
\end{subfigure}
\begin{subfigure}{.195\textwidth}
\includegraphics[width=\columnwidth]{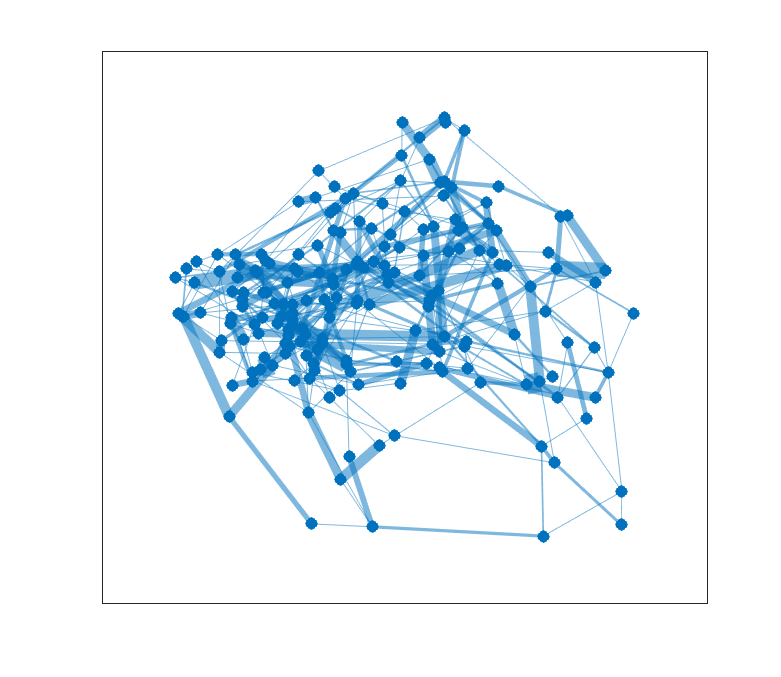}
\caption{$K=200$}\label{fig.grid.k200}
\end{subfigure}
\begin{subfigure}{.195\textwidth}
\includegraphics[width=\columnwidth]{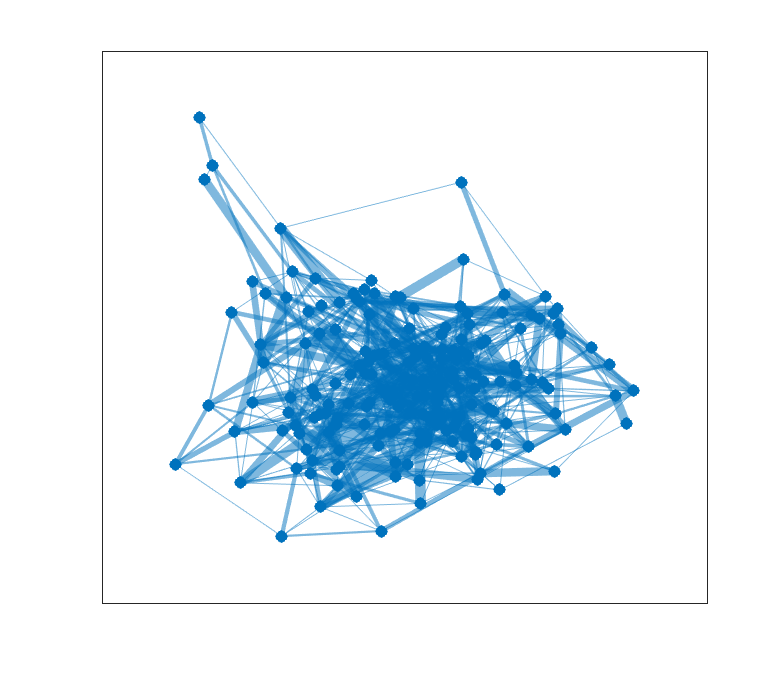}
\caption{$K=500$}\label{fig.grid.k500}
\end{subfigure}
  \caption{Network digarams for the Icelandic power grids with $K$ randomly added lines.}\label{fig.grid}
\end{figure*}

\section{Conclusions}\label{sec.concl}

In this paper we have analyzed power grid synchronization, with the explicit aim of reconciling system theoretic tools with power engineering insights and practice. With this motivation, we favored performance metrics based on the step response, and attempted to isolate two components: system wide frequency deviations characterized by time domain specifications, and  transient oscillatory deviations measured by an $\mathcal{L}_2$ norm.
Our analysis obtains such a decomposition for networks with heterogeneous machine ratings, provided a certain proportionality condition holds. Under this assumption, closed-form  performance metrics are given in terms of the network, the ratings, and the parameters of a representative machine. We also applied robustness tools to analyze deviations from the proportionality assumption, and presented simulation evidence on a real set of data to validate the relevance of our analysis.

A first observation is that machine models matter: to replicate observed behaviors such as overshoots in mean frequency response, requires enhancing classical swing equation model with a some representation of turbine controls. A less obvious conclusion is the rather minor role of system \emph{inertia} in these performance metrics: while inertia clearly influences the speed of initial ramping (RoCoF), its impact on the maximal deviation (Nadir) of the mean system frequency, and on the oscillation cost is very limited. Details depend on the machine model employed, but on the whole it appears that inertia is \emph{not} the main parameter on which synchronization performance rests. This should assuage some concerns (e.g.  \cite{8450880}) regarding the decreased inertia of power grids with high penetration of renewable energy sources.\footnote{{\color{black}{Of course, this analysis assumes \emph{some} inertia is still present; in the extreme case of a grid where all generation is connected by power electronics, other issues like voltage dynamics are likely to dominate.}}}

What is the key parameter, then? Both our theoretical results and our simulations point to the damping $d$ in the swing equation model as the main ``knob" to control synchronization
\footnote{{\color{black}{Consistently, the recent reference \cite{ulbig2014impact} that ostensibly focuses on low inertia ends up advocating for an "augmented frequency damping" as a solution.}}}.
The damping $d$ is mainly due to the response of the local \emph{load} to variations in frequency, which suggests that active load control is the best frequency regulation tool to mitigate step disturbances. Indeed, recent literature \cite{Zhao:2014bp,mallada2017optimal,guggilam2017engineering,pang2017optimal,7448474} has explored the potential of this method, also enabled by smarter grids. Our results lend support for this control strategy for improving synchronization in future power networks.

{\color{black}{If more active control is enabled, one could choose to go beyond the emulation of physical parameters such as damping or inertia, and use a more general dynamic controller; indeed such approaches are also being considered \cite{m2016cdc,jpm2017cdc}.
}}

\appendices


\section{Inner product computation}\label{app.ip}

We show here how to evaluate the inner product in \reeq{defy}
using state-space methods. We start with a minimal state-space
realization of the step response of the representative machine,
\begin{align*}
\tilde{h}_0(s) = \stsp{A}{B}{C}{0}.
\end{align*}
If follows easily that each $\tilde{h}_k$ in \reeq{htil} has realization
\begin{align*}
h_k(s) = \stsp{A_k}{B}{C}{0}, \mbox{ where } A_k = A - \lambda_k B C.
\end{align*}
Note that the state matrix $A_k$ is the only one that depends on
the eigenvalue $\lambda_k$ under consideration. Under Assumption
\ref{ass.strpassive}, $A_k$ is a Hurwitz matrix for any $\lambda_k>0$.

Writing $\tilde{h}_k(t) = C e^{A_k t} B$ for the impulse response, we compute the inner product between two such functions:
\begin{align}
\ip{\tilde{h}_k}{\tilde{h}_l} &= \int_0^\infty \tilde{h}_k(t) \tilde{h}_l(t)^T dt \nonumber  \\
&= C \underbrace{
\left(\int_0^\infty e^{A_k t} B B^T e^{A_l^T t} dt \right)}_{Q_{kl}} C^T; \label{eq.Qkl}
\end{align}
here $^T$ denotes matrix transpose. A standard calculation shows that
 $Q_{kl}$ satisfies the Sylvester equation
\begin{align}\label{eq.sylv}
A_k Q_{kl} + Q_{kl} A_l^T + BB^T = 0.
\end{align}
Furthermore since the eigenvalues of $A_k, A_l$ never add up to zero it follows (see
\cite{ZhouDG}) that \reeq{sylv} has a unique solution $Q_{kl}$.
Thus, the relevant inner product can be found by solving the above linear equation and
substituting into \reeq{Qkl}.

\subsection*{Second order machine model}

In the situation of Section \ref{sec.swing}, it is easily checked that
\[
A_k = \bmat{0 & 1 \\ -\frac{\lambda_k}{m} & -\frac{d}{m}} \quad B= \bmat{0 \\ -\frac{1}{m}}; \quad C = \bmat{1 & 0 }.
\]
In this case the solution to the Sylvester equation is
\begin{align*}
Q_{kl} = \frac{2 d}{m(\lambda_k - \lambda_l)^2 + 2(\lambda_k + \lambda_l)d^2}
\bmat{1 & \frac{\lambda_k - \lambda_l}{2d} \\ \frac{\lambda_l - \lambda_k}{2d}&  \frac{\lambda_k + \lambda_l}{2m}};
\end{align*}
Substitution into \reeq{Qkl} for the given $C$ proves Proposition
\ref{prop.ipswing}.

\subsection*{Third order machine model}

Here the relevant matrices are
\[ \setlength{\arraycolsep}{.3em}
A_k = \bmat{0 & 1 & 0\\-\frac{\lambda_k}{m} & -\frac{d}{m} & \frac{1}{m}\\
0 &  -\frac{r^{-1}}{\tau} &  -\frac{1}{\tau}};\
B = \bmat{0\\-\frac{1}{m}\\0};
\ C = \bmat{1 & 0 & 0}.
\]
The Sylvester equations for $Q_{kl}$ in this case (9 linear equations, 9 unknowns) give
unwieldy expressions. We report first the case $k=l$, which remains tractable; here
we have a  Lyapunov equation with symmetric solution
\begin{align}
Q_{kk} =
\frac{1}{\Delta} \bmat{\frac{m + \tau(\lambda_k + d)}{\lambda_k} & 0 & -\tau r^{-1} \\
0 & \frac{m+\tau(r^{-1}+ \lambda_k \tau + d)}{m} & -r^{-1}\\
-\tau r^{-1} &   r^{-1} & r^{-2}},
\end{align}
where $\Delta = 2[m(r^{-1}+d) +  \tau d (r^{-1}+d + \lambda_k\tau)].$
By looking at the $(1,1)$ element of this matrix we find the norm
$\|\tilde{h}_k\|^2 = \ip{\tilde{h}_k}{\tilde{h}_k}$, which coincides with the expression given in
\reeq{hknormturbine}.

Going now to the general case $k\neq l$, we report here only the inner product
obtained from the (1,1) element of $Q_{kl}$, itself found by solving the Sylvester equation
using the Matlab symbolic toolbox:

\begin{align*}
\langle{\tilde{h}_k},{\tilde{h}_l}\rangle= \frac{N}{D},
\end{align*}
where
\begin{align*}
N=&
2(2dm^2 + 2m^2r^{-1} + 2d^3\tau^2 + 4d^2m\tau + 2d^2\lambda_k\tau^3 \\
& + 2d^2\lambda_l\tau^3 + 2d^2r^{-1}\tau^2 + 4dmr^{-1}\tau + 2d\lambda_k\lambda_l\tau^4 \\
&+ 2d\lambda_km\tau^2 + 2d\lambda_lm\tau^2 + d\lambda_kr^{-1}\tau^3 + d\lambda_lr^{-1}\tau^3 \\
& + \lambda_kmr^{-1}\tau^2 + \lambda_lmr^{-1}\tau^2), \\
D =&
4d^4\lambda_k\tau^2 + 4d^4\lambda_l\tau^2 + 4d^3\lambda_k^2\tau^3 + 8d^3\lambda_k\lambda_l\tau^3 \\
& + 8d^3\lambda_km\tau + 8d^3\lambda_kr^{-1}\tau^2 + 4d^3\lambda_l^2\tau^3 + 8d^3\lambda_lm\tau \\
& + 8d^3\lambda_lr^{-1}\tau^2  + 4d^2\lambda_k^2\lambda_l\tau^4 + 6d^2\lambda_k^2m\tau^2 \\
& + 2d^2\lambda_k^2r^{-1}\tau^3 + 4d^2\lambda_k\lambda_l^2\tau^4 + 4d^2\lambda_k\lambda_lm\tau^2 \\
&+ 12d^2\lambda_k\lambda_lr^{-1}\tau^3  + 4d^2\lambda_km^2 + 16d^2\lambda_kmr^{-1}\tau\\
& + 4d^2\lambda_kr^{-2}\tau^2+ 6d^2\lambda_l^2m\tau^2 + 2d^2\lambda_l^2r^{-1}\tau^3 \\
&+ 4d^2\lambda_lm^2  + 16d^2\lambda_lmr^{-1}\tau + 4d^2\lambda_lr^{-2}\tau^2 +2d\lambda_k^3m\tau^3 \\
& - 2d\lambda_k^2\lambda_lm\tau^3 + 4d\lambda_k^2m^2\tau - 2d\lambda_k\lambda_l^2m\tau^3 - 8d\lambda_k\lambda_lm^2\tau \\
& + 16d\lambda_k\lambda_lmr^{-1}\tau^2  + 8d\lambda_km^2r^{-1}  +  8d\lambda_kmr^{-2}\tau \\
& + 2d\lambda_l^3m\tau^3 + 4d\lambda_l^2m^2\tau + 8d\lambda_lm^2r^{-1} + 8d\lambda_lmr^{-2}\tau \\ 
& + 2\lambda_k^3\lambda_lm\tau^4 + 2\lambda_k^3m^2\tau^2 + \lambda_k^3mr^{-1}\tau^3 - 4\lambda_k^2\lambda_l^2m\tau^4 \\
 &- 2\lambda_k^2\lambda_lm^2\tau^2 - \lambda_k^2\lambda_lmr^{-1}\tau^3  + 2\lambda_k^2m^3 - 2\lambda_k^2m^2r^{-1}\tau  \\
 &+ 2\lambda_k\lambda_l^3m\tau^4 - 2\lambda_k\lambda_l^2m^2\tau^2  - \lambda_k\lambda_l^2mr^{-1}\tau^3 \\
 &- 4\lambda_k\lambda_lm^3 + 4\lambda_k\lambda_lm^2r^{-1}\tau  + 4\lambda_km^2r^{-2} + 2\lambda_l^3m^2\tau^2 \\ &+ \lambda_l^3mr^{-1}\tau^3  + 2\lambda_l^2m^3 - 2\lambda_l^2m^2r^{-1}\tau+ 4\lambda_lm^2r^{-2}.
\end{align*}
The limiting cases $m\to 0$ and $m\to \infty$, presented in Section \ref{sec.turbine} were obtained from this general formula.

\section{Time-domain results with turbine control}
\label{app.nadirrocof}


\begin{IEEEproof}[Proof of Proposition \ref{prop.nadir}]

Since the second order system $g_0(s)$ is stable, the maximum of the impulse response is either at $t=0$ or the first time $\dot{g_0}(t)=0$. Since $\bar g_0(0)=0$, then it must be the latter.

Therefore we need to find the time instance $t_\text{Nadir}$ such that $|g_0(t_\text{Nadir})|=\underset{t\geq0}{\sup}|g_0(t)|$.

The time derivative of \eqref{eq:g0-turbine} is given by
\begin{align}
\dot g_0(t) &=\mathcal L^{-1}\left\{ sg_0(s)-g_0(t)|_{t=0^+}\right\} \\
&= \frac{1}{m}\sqrt{1+(\tan(\phi))^2}e^{-\eta t}
\cos(\omega_{d}t-\phi)~\label{eq:dot-g0-turbine}
\end{align}
where $\phi$ is defined as in the \eqref{eq:sinphi}.

Setting now $\dot g_0(t)=0$ in \eqref{eq:dot-g0-turbine} gives
\[
t_k =  \frac{\phi +\frac{\pi}{2}+k\pi}{\omega_d}, \quad k\geq0
\]
and since $\phi+\frac{\pi}{2}>0$, the first maximum is for $k=0$.
Therefore $$t_\text{Nadir}=\frac{\phi+\frac{\pi}{2}}{\omega_d}$$
which after substituting in \eqref{eq:g0-turbine} gives
\begin{align*}
&||g_0||_\infty = |g_0(t_\text{Nadir})| 
&\!=\!\frac{1}{d\!+\!r^{-1}}\left(1+\sqrt{\frac{\tau r^{-1}}{m} }e^{\!-\!\frac{\eta}{\omega_d}\left(\!\phi+\frac{\pi}{2}\!\right)}\right)
\end{align*}
where the last step follows from \eqref{eq:omegad}, \eqref{eq:eta-gamma}, \eqref{eq:sinphi} and
\begin{equation}\label{eq:cosphi}
\cos(\phi) = \frac{\omega_d}{ \sqrt{\omega_d^2 + \left(\frac{1}{\tau}-\eta\right)^2 } }=\sqrt{1-\frac{(m-d\tau)^2}{4m\tau r^{-1}}},
\end{equation}
as well as several trigonometric identities.

Therefore the Nadir of the system frequency is given by
\begin{equation}\label{eq:nadir-turbine}
||\overline{w}||_\infty = \frac{\left|\sum_i u_i \right|}{\sum_i f_i}\frac{1}{d+r^{-1}}\left(1+\sqrt{\frac{\tau r^{-1}}{m} }e^{\!-\!\frac{\eta}{\omega_d}\left(\!\phi+\frac{\pi}{2}\!\right)}\right)
\end{equation}
\end{IEEEproof}

\subsection*{Proof of Proposition \ref{prop.dNadirdm}}\label{app.nadir.2}

The proof is based on two previous lemmas.

\begin{lemma}\label{lem:dphidm}
Given a power system under Assumption \ref{ass.scale} with $g_i(s)$ given by~\eqref{eq.giturbine}, the derivative of $\phi$ with respect to $m$ is given by
\[
\frac{\partial \phi}{\partial m}=\frac{m+d\tau}{2m\sqrt{4r^{-1}m\tau-(m-d\tau)^2}}
\]
\end{lemma}
\begin{IEEEproof}
Notice that while it is not possible to derive a closed form condition for $\phi$ in terms of the system parameters without the aid of a trigonometric function, it is possible to achieve such expression for $\frac{\partial \phi}{\partial m}$. This is achieved by first computing
\begin{align*}
\frac{\partial}{\partial m}\tan(\phi(m))&=\frac{\partial}{\partial m}\left( \frac{(m\!-\!d\tau)}{\sqrt{4r^{\!-\!1}m\tau\!-\!(m\!-\!d\tau)^2}} \right) \\
&=\frac{2r^{\!-\!1}\tau(m+d\tau)}{(4r^{\!-\!1}m\tau\!-\!(m\!-\!d\tau)^2)^\frac{3}{2}}
\end{align*}

Now using the fact that $\phi\in(-\frac{\pi}{2},\frac{\pi}{2})$ we get
\begin{align*}
\frac{\partial }{\partial m}\phi(m) &= \frac{\partial}{\partial m}\arctan\left(\tan (\phi(m))\right)\\
&=\left(\frac{\partial}{\partial x}\arctan(x)\big|_{x=\tan(\phi)}\right)
\frac{\partial}{\partial m}\tan(\phi(m))\\
&=\frac{m+d\tau}{2m\sqrt{4r^{-1}m\tau-(m-d\tau)^2}}
\end{align*}
\end{IEEEproof}

\begin{lemma}\label{lem:detaomega-dm}
Given a power system under Assumption \ref{ass.scale} with $g_i(s)$ given by~\eqref{eq.giturbine}, the derivative of $\frac{\eta}{\omega_d}$ with respect to $m$ is given by
\[
\frac{\partial}{\partial m}\left( \frac{\eta}{\omega_d}\right)
=
\frac{
2\tau(d+ r^{-1})(m-d\tau)
}{
({4m\tau r^{-1}-\left(m-d\tau\right)^2})^\frac{3}{2}
}
\]
\end{lemma}

\begin{IEEEproof}
The proof is just by direct computation
\begin{align*}
&\frac{\partial}{\partial m}\left( \frac{\eta}{\omega_d}\right)
=\frac{\partial}{\partial m}\left(\frac{m+d\tau}{\sqrt{4m\tau r^{\!-\!1}\!-\!\left(m\!-\!d\tau\right)^2}}\right)\\
&=
\frac{
{4m\tau r^{\!-\!1}\!-\!\left(m\!-\!d\tau\right)^2}
\!-\!(2\tau r^{\!-\!1}\!-\!m+d\tau)(m+d\tau)
}{
\sqrt{(4m\tau r^{\!-\!1}\!-\!\left(m\!-\!d\tau\right)^2)}({4m\tau r^{\!-\!1}\!-\!\left(m\!-\!d\tau\right)^2})
}\\
&=
\frac{
2\tau(d+ r^{\!-\!1})(m\!-\!d\tau)
}{
({4m\tau r^{\!-\!1}\!-\!\left(m\!-\!d\tau\right)^2})^\frac{3}{2}
}
\end{align*}
\end{IEEEproof}

\begin{IEEEproof}[Proof of Proposition \ref{prop.dNadirdm}]
Using lemmas \ref{lem:dphidm} and \ref{lem:detaomega-dm}, and $\alpha=\frac{\pi}{2}+\phi$ we can compute
\begin{align*}
&\frac{\partial }{\partial m}\!\left[\!\left(\!\frac{w_\infty}{\sqrt{\tau r^{\!-\!1}}}\!\right)^{-1}\!\!\!||\tilde h_0||_\infty\!\right]=
\frac{\partial }{\partial m}\left[
\left(\frac{1}{\sqrt{\tau r^{\!-\!1}}}\!+\!
\frac{1}{\sqrt{m}}
e^{\!-\!\frac{\eta}{\omega_d}\alpha}\right)
\right]\\
&=e^{\!-\!\frac{\eta}{\omega_d}\alpha}
\left(
\frac{\partial}{\partial m}\left(\frac{1}{\sqrt{m}}\right)
\!-\!
\frac{1}{\sqrt{m}}
\left(
\frac{\eta}{\omega_d}\frac{\partial \phi}{\partial m} \!+\!\alpha\frac{\partial}{\partial m}\left(\frac{\eta}{\omega_d}\right)
\right)
\right)\\
&=\frac{e^{\!-\!\frac{\eta}{\omega_d}\alpha}}{2m\sqrt{m}}
\Bigg(\!-\!1 \!-\!
\frac{(m\!+\!d\tau)^2}
{{4m\tau r^{\!-\!1}\!-\!(m\!-\!d\tau)^2}}\\
&\quad 
\!-\!\alpha
\frac{4m\tau(d\!+\! r^{\!-\!1})(m\!-\!d\tau)}{({4m\tau r^{\!-\!1}\!-\!(m\!-\!d\tau)^2})^\frac{3}{2}}
\Bigg).
\end{align*}
It is easy to see that whenever  $m\geq d\tau$, then $\frac{\partial}{\partial m} \bar w<0$. However it is possible to show that even when $m<d\tau$, the remaining negative terms are still dominant making $\frac{\partial}{\partial m} \bar w<0$ always true.
\end{IEEEproof}

\subsection*{Proof of Proposition \ref{prop.rocof}}\label{app.rocof}

Again we first state a lemma with some of the calculations.
\begin{lemma}
Given $g_0(s)$ as in \eqref{eq.giturbine}
\[
\ddot{g}_0(t) = -\frac{d}{m^2}e^{-\eta t}\frac{\cos(\omega_d t-\beta)}{\cos(\beta)}
\]
where the angle $\beta\in(-\frac{\pi}{2},\frac{\pi}{2})$, $\beta>\phi$ for $\phi\in(-\frac{\pi}{2},\frac{\pi}{2})$ defined by \eqref{eq:cosphi} and \eqref{eq:sinphi},  and $\beta$ is uniquely defined by
\[
\cos(\beta) =\frac{\omega_d}{ \sqrt{\omega_d^2+\left(\frac{1}{\tau}-\eta+\frac{r^{-1}}{d\tau}\right)^2} }
\]
and
\[
\sin(\beta) =\frac{\frac{1}{\tau}-\eta+\frac{r^{-1}}{d\tau}}{ \sqrt{\omega_d^2+\left(\frac{1}{\tau}-\eta+\frac{r^{-1}}{d\tau}\right)^2} }.
\]
\end{lemma}
\begin{IEEEproof}
We first compute
\begin{align}
s^2g_0(s)&=\frac{1}{m\tau}\frac{\tau s^2 + s}{s^2+\left(\frac{1}{\tau}+\frac{d}{m}\right)s+\frac{d+r^{-1}}{m\tau}}\\
&=\frac{1}{m\tau}\left(\tau - \frac{\frac{d\tau}{m}s+\frac{d+r^{-1}}{m}}{s^2+\left(\frac{1}{\tau}+\frac{d}{m}\right)s+\frac{d+r^{-1}}{m\tau}}\right)\\
&=\frac{1}{m\tau}\left(\tau - h(s)\right)
\end{align}
where
\[
h(s) =\frac{\frac{d\tau}{m}s+\frac{d+r^{-1}}{m}}{s^2+\left(\frac{1}{\tau}+\frac{d}{m}\right)s+\frac{d+r^{-1}}{m\tau}}.
\]

Thus we can compute
\begin{align*}
\ddot g_0(t)&=\mathcal L^{-1}\left[ s^2g_0(s)-sg_0(t)|_{t=0^+}-\dot g_0(t)|_{t=0^+}\right]\\
&=\mathcal L^{-1}\left[ \frac{1}{m\tau}(\tau-h(s)) - \frac{1}{m}\right]\\
&=-\frac{1}{m\tau}\mathcal L^{-1}\left[h(s)\right]
\end{align*}

Therefore, it is enough to compute
\begin{align*}
&h(t)=\mathcal L^{\!-\!1}\left[ h(s)\right]\\
&= \mathcal L^{\!-\!1}\left[
\frac{\frac{d\tau}{m}s\!+\!\frac{d\!+\!r^{\!-\!1}}{m}}{ (s\!+\!\eta)^2 \!+\!\omega_d^2}
\right]\\
&=e^{\!-\!\eta t}\left(\frac{d\tau}{m}\cos(\omega_d t) \!+\! \frac{\frac{d\!+\!r^{\!-\!1}}{m}\!-\!\eta\frac{d\tau}{m}}{\omega_d}\sin(\omega_d t)\right)\\
&=\frac{d\tau}{m}e^{\!-\!\eta t}\frac{\cos(\omega_d t\!-\!\beta)}{\cos(\beta)}
\end{align*}
\end{IEEEproof}


\begin{IEEEproof}[Proof of Proposition \ref{prop.rocof}]
When $t=0^+$, $\ddot g_0(0)<0$. Therefore, the initial trend the $\dot{g_0}(t)$ is decreasing and therefore
\[
\dot{g_0}(0)=
\frac{1}{m}\sqrt{1+(\tan(\phi))^2}e^{0}\cos(0-\phi) = \frac{1}{m}
\]
is a local maximum.

We will now show that this is indeed a global maximum.
The function $\ddot g_0(t)$ crosses zero for the first time when $\omega_d t-\beta=\frac{\pi}{2}$ or equivalently
\begin{equation}\label{eq:t_star}
t^* = \frac{\beta+\frac{\pi}{2}}{\omega_d}
\end{equation}

By substituting \eqref{eq:t_star} into \eqref{eq:dot-g0-turbine} and
defining
\begin{align*}
\Delta&=\sqrt{\omega_d^2 + \left(\frac{1}{\tau}\!-\!\eta+\frac{r^{-1}}{d\tau}\right)^2}
\end{align*}
we get
\begin{align*}
&\dot{g}_0(t^*) = \frac{1}{m}\sqrt{1+(\tan(\phi))^2}e^{-\eta t^*}
\cos(\omega_dt^*-\phi)\\
&=\frac{1}{m}e^{-\frac{\eta}{\omega_d}(\beta +\frac{\pi}{2})}\frac{\cos(\frac{\pi}{2}+(\beta-\phi))}{\cos(\phi)}\\
&=-\frac{1}{m}e^{-\frac{\eta}{\omega_d}(\beta +\frac{\pi}{2})}
\left(
\frac{\frac{r^{-1}}{d\tau}}
{\Delta}
\right)
\end{align*}
Finally by simplifying
\begin{align*}
\Delta&=\sqrt{\omega_d^2 + \left(\frac{1}{\tau}\!-\!\eta+\frac{r^{-1}}{d\tau}\right)^2}
={\sqrt{ \frac{r^{-1}}{d\tau}\frac{1}{\tau}+     \left(\frac{r^{-1}}{d\tau}\right)^2}}
\end{align*}
we obtain
\begin{align*}
&\dot{g}_0(t^*) =-\frac{1}{m}e^{-\frac{\eta}{\omega_d}(\beta +\frac{\pi}{2})}
\left(
\frac{\frac{r^{-1}}{d\tau}}
{\sqrt{ \frac{r^{-1}}{d\tau}\frac{1}{\tau}+     \left(\frac{r^{-1}}{d\tau}\right)^2}}
\right)>-\frac{1}{m}
\end{align*}
Therefore, since for any additional instant of time that $\dot{g}_0(t)$ achieves a local extremum $t_k^*=\frac{\beta + \frac{\pi}{2} +k\pi}{\omega_d}$, the factor $|\cos(\omega_dt_k^*-\phi)|$ does not change, then the maximum is achieved at $t=0$.

\end{IEEEproof}

\bibliographystyle{ieeetr}
\bibliography{refs,addnl-bib}

\begin{IEEEbiography}[{\includegraphics[width=1in,height
=1.25in,clip,keepaspectratio]{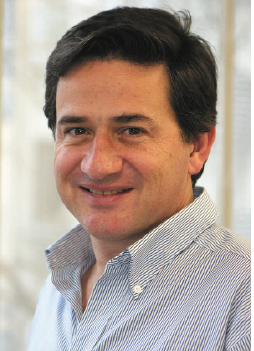}}]{Fernando Paganini} (M'90--SM'05--F'14)
received his Electrical Engineering and Mathematics degrees from
Universidad de la Rep\'ublica, Montevideo, Uruguay, in 1990, and his
M.S. and PhD degrees in Electrical Engineering from the California
Institute of Technology, Pasadena, in 1992 and 1996 respectively.
His PhD thesis received the 1996 Wilts Prize and the 1996 Clauser
Prize at Caltech. From 1996 to 1997 he was a postdoctoral associate at MIT. Between
1997 and 2005 he was on the faculty the Electrical Engineering
Department at UCLA, reaching the rank of Associate Professor.
Since 2005 he is Professor of Electrical and Telecommunications
Engineering at Universidad ORT Uruguay. Dr. Paganini has received the 1995 O. Hugo Schuck Best Paper
Award, the 1999 NSF CAREER Award, the 1999 Packard Fellowship, the
2004 George S. Axelby Best Paper Award, and the 2010 Elsevier
Scopus Prize. He is a member of the Uruguayan National Academy of
Sciences, the Uruguayan National Academy of Engineering, and the Latin American Academy of Sciences. He is a Fellow of the IEEE. His research interests are control and networks.
\end{IEEEbiography}


\begin{IEEEbiography}[{\includegraphics[width=1in,height=1.25in,clip,keepaspectratio]{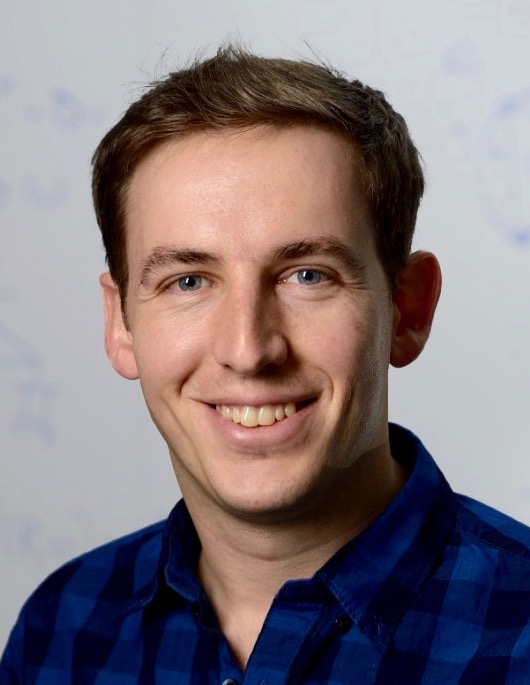}}]
{Enrique Mallada} (S'09-M'13) is an Assistant Professor of Electrical and Computer Engineering at Johns Hopkins University. Prior to joining Hopkins in 2016, he was a Post-Doctoral Fellow in the Center for the Mathematics of Information at Caltech from 2014 to 2016. He received his Ingeniero en Telecomunicaciones degree from Universidad ORT, Uruguay, in 2005 and his Ph.D. degree in Electrical and Computer Engineering with a minor in Applied Mathematics from Cornell University in 2014. 
Dr. Mallada was awarded 
the NSF CAREER award in 2018,
the ECE Director's PhD Thesis Research Award for his dissertation in 2014, 
the Center for the Mathematics of Information (CMI) Fellowship from Caltech in 2014,
and the Cornell University Jacobs Fellowship in 2011. 
His research interests lie in the areas of control, dynamical systems and optimization, with applications to engineering networks such as power systems and the Internet.
\end{IEEEbiography}

\end{document}